\pdfoutput=1

\documentclass[a4paper,USenglish]{lipics}
\usepackage{microtype}%if unwanted, comment out or use option "draft"

\usepackage{hyperref,xcolor}

\definecolor{darkred}  {rgb}{0.5,0,0}
\definecolor{darkblue} {rgb}{0,0,0.5}
\definecolor{darkgreen}{rgb}{0,0.5,0}
\hypersetup{
  pdftitle = {Easy and hard functions for the Boolean hidden shift problem},
  pdfauthor = {Andrew M. Childs, Robin Kothari, Maris Ozols, Martin Roetteler},
  colorlinks = true,
  urlcolor  = darkblue,     % color of external links
  linkcolor = darkblue,     % color of internal links
  citecolor = darkgreen,    % color of links to bibliography
  filecolor = darkred       % color of file links
}

\title{Easy and hard functions for the Boolean hidden shift problem}

\author[1]{Andrew M.~Childs}
\author[2]{Robin Kothari}
\author[3(1,4)]{Maris Ozols}
\author[4]{Martin Roetteler}
\affil[1]{Department of Combinatorics \& Optimization and\\
Institute for Quantum Computing, University of Waterloo\\
200 University Avenue West,
Waterloo, ON, N2L 3G1, Canada\\
\texttt{amchilds@uwaterloo.ca}}
\affil[2]{David R.\ Cheriton School of Computer Science and\\
Institute for Quantum Computing, University of Waterloo\\
200 University Avenue West,
Waterloo, ON, N2L 3G1, Canada\\
\texttt{rkothari@uwaterloo.ca}}
\affil[3]{IBM TJ Watson Research Center\\
1101 Kitchawan Road, Yorktown Heights, NY 10598, USA\\
\texttt{marozols@yahoo.com}}
\affil[4]{NEC Laboratories America\\
4 Independence Way, Suite 200, Princeton, NJ 08540, USA\\
\texttt{mroetteler@nec-labs.com}}
\authorrunning{A.~M.~Childs, R.~Kothari, M.~Ozols, and M.~Roetteler}

\usepackage{amsmath,amssymb,amsfonts,amsthm,amstext}

\usepackage{thm-restate}

\usepackage{mathtools}
\mathtoolsset{centercolon} % center the colon in "\defeq"

\usepackage{tikz}

\usepackage{colonequals}

\theoremstyle{definition}

\newtheorem{problem}{Problem}
\newtheorem{proposition}[theorem]{Proposition}
\newtheorem*{fact*}{Fact}
\newtheorem*{example*}{Example}

\newcommand{\ket}[1]{|#1\rangle}
\newcommand{\bra}[1]{\langle#1|}
\newcommand{\braket}[2]{\langle#1|#2\rangle}
\newcommand{\proj}[1]{|#1\rangle\langle#1|}

\newcommand{\x}{\otimes}
\newcommand{\xp}[1]{^{\otimes #1}}
\newcommand{\op}{\oplus}

\newcommand{\tp}{^{\mathsf{T}}}

\newcommand{\vc}[1]{\boldsymbol{#1}} % real vector
\newcommand{\bv}[1]{#1}%{\vec{#1}}   % Boolean vector

\DeclarePairedDelimiter{\set}{\lbrace}{\rbrace}
\DeclarePairedDelimiter{\abs}{\lvert}{\rvert}
\DeclarePairedDelimiter{\norm}{\lVert}{\rVert}

\newcommand{\C}{\mathbb{C}}
\newcommand{\R}{\mathbb{R}}

\newcommand{\Z}{\mathbb{Z}}

\DeclareMathOperator{\NOT}{NOT}
\DeclareMathOperator{\AND}{AND}
\DeclareMathOperator{\OR}{OR}

\newcommand{\Thm}[1]{\hyperref[thm:#1]{Theorem~\ref*{thm:#1}}}
\newcommand{\Lem}[1]{\hyperref[lem:#1]{Lemma~\ref*{lem:#1}}}
\newcommand{\Cor}[1]{\hyperref[cor:#1]{Corollary~\ref*{cor:#1}}}
\newcommand{\Def}[1]{\hyperref[def:#1]{Definition~\ref*{def:#1}}}
\newcommand{\Prop}[1]{\hyperref[prop:#1]{Prop.~\ref*{prop:#1}}}
\newcommand{\Prob}[1]{\hyperref[prob:#1]{Problem~\ref*{prob:#1}}}

\newcommand{\Sect}[1]{\hyperref[sect:#1]{Sect.~\ref*{sect:#1}}}
\newcommand{\Apx}[1]{\hyperref[apx:#1]{Appendix~\ref*{apx:#1}}}

\newcommand{\Fig}[1]{\hyperref[fig:#1]{Fig.~\ref*{fig:#1}}}
\newcommand{\Tab}[1]{\hyperref[tab:#1]{Table~\ref*{tab:#1}}}

\newcommand{\EqRef}[1]{\hyperref[eq:#1]{(\ref*{eq:#1})}}
\newcommand{\Eq}[1]{Eq.~\hyperref[eq:#1]{(\ref*{eq:#1})}}
\newcommand{\step}[1]{\hyperref[step:#1]{step~\ref*{step:#1}}}

\newcommand{\alg}[1]{\textnormal{\textbf{#1}}}
\newcommand{\PGMtxt}{\alg{PGM}}
\newcommand{\PGM}{\hyperref[alg:PGM]{\PGMtxt}}
\newcommand{\BHSP}{\textsc{BHSP}}

\newcommand{\quantump}{\pi}
\newcommand{\quantums}{\sigma}
\newcommand{\ve}{\vc{\varepsilon}}
\newcommand{\waterfilling}{\ve_{\vc{\quantump}\to\vc{\quantums}}^p}

\newcommand{\mc}[1]{\mathcal{#1}}
\newcommand{\E}{\mathbb{E}}
\newcommand{\defeq}{\colonequals}

\newcommand{\I}[1]{\mathrm{I}_{#1}}

\newenvironment{algobox}[1]
{\begin{center}\begin{minipage}{#1\textwidth}\hrulefill\\}
{\vspace{-8pt}\hrulefill\end{minipage}\end{center}}

\newcommand{\FF}{\mc{F}}           %
\newcommand{\FS}[2]{\FF^{#1}_{#2}} % t-fold Fourier state
\newcommand{\FC}[1]{\FF^{#1}}      % t-fold Fourier coefficient
\usepackage{color}
\usepackage[bordercolor=white,backgroundcolor=gray!30,linecolor=black,colorinlistoftodos]{todonotes}

\renewcommand{\copyrightline}{} % suppress LIPICS logo for arxiv version

\begin{document}

\maketitle

\begin{abstract}
We study the quantum query complexity of the Boolean hidden shift problem.  Given oracle access to $f(x+s)$ for a known Boolean function $f$, the task is to determine the $n$-bit string $s$. The quantum query complexity of this problem depends strongly on $f$. We demonstrate that the easiest instances of this problem correspond to bent functions, in the sense that an exact one-query algorithm exists if and only if the function is bent. We partially characterize the hardest instances, which include delta functions. Moreover, we show that the problem is easy for random functions, since two queries suffice. Our algorithm for random functions is based on performing the pretty good measurement on several copies of a certain state; its analysis relies on the Fourier transform. We also use this approach to improve the quantum rejection sampling approach to the Boolean hidden shift problem.
\end{abstract}

%%%%%%%%%%%%%%%%%%%%%%
\section{Introduction} \label{sect:Intro}
%%%%%%%%%%%%%%%%%%%%%%

Many computational problems for which quantum algorithms can achieve superpolynomial speedup over the best known classical algorithms are related to the \emph{hidden subgroup problem} (see for example \cite{CvD10}).

\begin{problem}[Hidden subgroup problem]\label{prob:Subgroup}
For any finite group $G$, say that a function $f\colon G \to X$ \emph{hides} a subgroup $H$ of $G$ if it is constant on cosets of $H$ in $G$ and distinct on different cosets. Given oracle access to such an $f$, find a generating set for $H$.
\end{problem}

Two early examples of algorithms for hidden subgroup problems are the Deutsch--Jozsa algorithm~\cite{DeutschJozsa} and Simon's algorithm~\cite{Simon}. Inspired by the latter, Shor discovered efficient quantum algorithms for factoring integers and computing discrete logarithms~\cite{Factoring}. Kitaev subsequently introduced the Abelian stabilizer problem and derived an efficient quantum algorithm for it that includes Shor's factoring and discrete logarithm algorithms as special cases~\cite{Kitaev95}. Eventually it was observed that all of the above algorithms solve special instances of the hidden subgroup problem~\cite{Jozsa98, MoscaEkert, Jozsa01}.

This early success created significant interest in studying various instances of the hidden subgroup problem and led to discovery of many other quantum algorithms. For example, period finding over the reals was used by Hallgren to construct an efficient quantum algorithm for solving Pell's equation~\cite{Hallgren}. Moreover, the hidden subgroup problem over symmetric and dihedral groups are related to the graph isomorphism problem~\cite{BonehLipton, Beals, Hoyer, EttingerHoyer} and certain lattice problems~\cite{Regev:2004}, respectively. The possibility of efficient quantum algorithms for these problems remains a major open question. Kuperberg has provided a subexponential-time quantum algorithm for the dihedral subgroup problem~\cite{Kuperberg, Regev, Kuperberg2}, which has been used to construct elliptic curve isogenies in quantum subexponential time~\cite{CJS10}.

The \emph{hidden shift problem} (also known as the \emph{hidden translation problem}) is a natural variant of the hidden subgroup problem.

\begin{problem}[Hidden shift problem]\label{prob:Shift}
Let $G$ be a finite group. Given oracle access to functions $f_0, f_1\colon G \to X$ with the promise that $f_0(x) = f_1(x \cdot s)$ for some $s \in G$, determine $s$.
\end{problem}

If $G$ is Abelian and $f_0$ is injective, this problem is equivalent to the hidden subgroup problem in the semidirect product group $G \rtimes \Z_2$, where the group operation is defined by
$(x_1, b_1) \cdot (x_2, b_2) \defeq \bigl( x_1 \cdot x_2^{(-1)^{b_1}}, b_1 + b_2 \bigr)$ and the hiding function $f\colon G \rtimes \Z_2 \to X$ is defined as $f[(x,b)] \defeq f_b(x)$. One can check that $f$ is constant on cosets of $H \defeq \langle(s,1)\rangle$ and that injectivity of $f_0$ implies that $f$ is distinct on different cosets. Thus, $f$ hides the subgroup $H$ in $G \rtimes \Z_2$.

Notice that if $G = \Z_d$ then $G \rtimes \Z_2$ is the dihedral group. Ettinger and H{\o}yer~\cite{EH00} showed that the dihedral hidden subgroup problem reduces to the special case of a subgroup $\langle(s,1)\rangle$. Thus the hidden shift problem in $\Z_d$ (with $f_0$ injective) is equivalent to the dihedral hidden subgroup problem, motivating further study of the hidden shift problem for various groups~\cite{vDHI:2003, Friedl2003, MRRS, CW07, CvD07, Ivanyos:2008}.

While the case where $f_0$ is injective is simply related to the hidden subgroup problem, one can also consider the hidden shift problem without this promise.  For example, van~Dam, Hallgren, and Ip~\cite{vDHI:2003} gave an efficient quantum algorithm to solve the shifted Legendre symbol problem, a non-injective hidden shift problem.  Their result breaks a proposed pseudorandom function~\cite{Damgard88}, showing the potential for cryptographic applications of hidden shift problems.  Work on hidden shift problems can also inspire new algorithmic techniques, such as quantum rejection sampling \cite{QRS}.  Moreover, negative results could have applications to designing classical cryptosystems that are secure against quantum attacks~\cite{Regev:2004}.

For the rest of the paper we restrict our attention to the \emph{Boolean hidden shift problem}, in which the hiding function has the form $f_0\colon \Z_2^n \to \Z_2$ for some integer $n \geq 1$.  For this problem (with $n>1$), $f_0$ is necessarily non-injective. This problem has previously been studied in~\cite{Roetteler:2009, Roetteler:2010, GRR11, QRS, Gharibi}.

Notice that to determine the hidden shift of an injective function $f_0$, it suffices to find $x_0$ and $x_1$ such that $f_0(x_0) = f_1(x_1)$. However, this does not hold in the non-injective case, so it is nontrivial to verify a candidate hidden shift (see \cite[Appendix~B]{QRS}). In fact, sometimes the hidden shift cannot be uniquely determined in principle (see \Sect{b-shifts}). On the other hand, by considering functions with codomain $\Z_2$, we have more structure than in the hidden subgroup problem or the injective hidden shift problem, where the codomain is arbitrary. We exploit this structure by encoding the values of the function as phases and using the Fourier transform.

More precisely, the main problem studied in this paper, sometimes denoted $\BHSP_f$, is as follows.

\begin{problem}[Boolean hidden shift problem]\label{prob:ShiftZd}
Given a complete description of a function $f\colon \Z_2^n \to \Z_2$ and access to an oracle for the \emph{shifted function} $f_s(x) \defeq f(x + s)$, 
determine the \emph{hidden shift} $s \in \Z_2^n$.
\end{problem}

Note that in degenerate cases, when the oracle does not contain enough information to completely recover the hidden shift, no algorithm can succeed with certainty.

Let us highlight the main differences between the above problem and other types of hidden shift problem. In the Boolean hidden shift problem,
\begin{itemize}
  \item the function $f$ is \emph{not} injective, and
  \item we are given a \emph{complete description} of the unshifted function $f$ instead of having only oracle access to $f$.
\end{itemize}
Moreover, we are interested only in the \emph{query complexity} of the problem and do not consider its time complexity. This means that we can pre-process the description of $f$ (which may be exponentially large) at no cost before we start querying the oracle.

This problem has been considered previously, e.g., by~\cite{QRS}. Note that some prior work does not give complete description of $f$ but only oracle access to it \cite{Roetteler:2009, Roetteler:2010, GRR11, Gharibi} (and in some cases~\cite{Roetteler:2010} also gives oracle access also to $\tilde{f}$, the dual bent function of $f$).

To address this problem on a quantum computer, we use an oracle that computes the shifted function in the phase.  Such an oracle can be implemented using only one query to an oracle that computes the function in a register.

\begin{definition}\label{def:Oracle}
The quantum \emph{phase oracle} is $O_{f_{s}}\colon \ket{x} \mapsto (-1)^{f(x+s)} \ket{x}$.
\end{definition}

More generally, one can use a controlled phase oracle $\bar O_{f_{s}}\colon \ket{b,x} \mapsto (-1)^{bf(x+s)} \ket{b,x}$ for $b \in \{0,1\}$, which is equivalent to an oracle that computes the function in the first register up to a Hadamard transform.  Some of our algorithms do not make use of this freedom, although our lower bounds always take it into account.

Ultimately, we would like to characterize the classical and quantum query complexities of the hidden shift problem for any Boolean function (or more generally, for any function $f\colon \Z_d^n \to \Z_d$). While we do not resolve this question completely, we make progress by providing a new quantum query algorithm (see \Sect{PGM}) and improving an existing one (see \Sect{QRS with t queries}). However, it remains an open problem to better understand both the classical and quantum query complexities of the $\BHSP$ for general functions.

While general functions are difficult to handle, the quantum query complexity of the hidden shift problem is known for two extreme classes of Boolean functions:
\begin{itemize}
  \item If $f$ is a \emph{bent function}, i.e., it has a ``flat'' Fourier spectrum (see \Sect{Easy = bent}), then one quantum query suffices to solve the problem exactly~\cite{Roetteler:2010}.
  \item If $f$ is a \emph{delta function}, i.e., $f(x) \defeq \delta_{x,x_0}$ for some $x_0 \in \Z_2^n$, then the hidden shift problem for $f$ is equivalent to unstructured search---finding $x_0 + s$ among the $2^n$ elements of $\Z_2^n$---so the quantum query complexity is $\Theta(\sqrt{2^n})$~\cite{Grover96, bbbv97}.
\end{itemize}
Intuitively, other Boolean functions should lie somewhere between these two extreme cases.  In this paper, we give formal evidence for this: we show that the problem can be solved exactly with one query only if $f$ is bent, and we show that it can be solved for any function with $O(\sqrt{2^n})$ queries, with a lower bound of $\Omega(\sqrt{2^n})$ only if the truth table of $f$ has Hamming weight $\Theta(1)$ or $\Theta(2^n)$. This is similar to the weighing matrix problem considered by van~Dam~\cite{WeighingMatrices}, which also interpolates between two extreme cases: the Bernstein-Vazirani problem~\cite{BV:97} and Grover search~\cite{Grover96}.

Aside from delta and bent functions, the Boolean hidden shift problem has previously been considered for several other families of functions. Boolean functions that are quadratic forms or are close to being quadratic are studied in~\cite{Roetteler:2009}. Random Boolean functions have been considered in~\cite{GRR11, Gharibi}. Finally, \cite{QRS} uses quantum rejection sampling to solve the $\BHSP$ for any function,
although its performance in general is not well understood.

Apart from algorithms designed specifically for the $\BHSP$, there are generic classical and quantum algorithms for the $\BHSP$ derived from learning theory. In particular, the $\BHSP$ can be viewed as an instantiation of the problem of exact learning through membership queries. The resulting algorithms are optimal for classical and quantum query complexity up to polynomial factors in $n$. More precisely, for any learning problem, Servedio and Gortler define a combinatorial parameter $\gamma$ \cite{SG04}. For the problem $\BHSP_f$, we denote the parameter as $\gamma_f$. From their results it follows that the classical query complexity of $\BHSP_f$ is lower bounded by $\Omega(n)$ and $\Omega(1/\gamma_f)$ and upper bounded by $O(n/\gamma_f)$. For quantum algorithms, they show a lower bound of $\Omega(1/\sqrt{\gamma_f})$. At{\i}c{\i} and Servedio \cite{AS05} later showed an upper bound of $O(n\log n /\sqrt{\gamma_f})$ queries.

The rest of this paper is organized as follows. In \Sect{QFT and convolution} we briefly review some basic Fourier analysis to establish notation. Next, in \Sect{Extremes} we explore the extreme cases of the $\BHSP$. In \Sect{PGM} we introduce a new approach to the $\BHSP$ based on the pretty good measurement. We analyze its performance for delta, bent, and random Boolean functions in \Sect{PGM performance}. In \Sect{QRS with t queries} we propose an alternative method for boosting the success probability of the quantum rejection sampling algorithm from~\cite{QRS}. Finally, \Sect{Conclusions} 
presents conclusions and open questions.

This paper has several appendices. In \Apx{Bent} we show that the easy instances of the $\BHSP$ correspond to bent functions. In \Apx{Random1}, we show that with one quantum query we can succeed on a constant fraction of all functions, whereas in \Apx{Random functions} we prove that two quantum queries suffice to solve the $\BHSP$ for random functions. Finally, in \Apx{Zeroes} we analyze the structure of zero Fourier coefficients of Boolean functions.

%--------------------------------------------%
\section{Fourier analysis} \label{sect:QFT and convolution}
%--------------------------------------------%

Our main tool is Fourier analysis of Boolean functions~\cite{DeWolf:2008}. Here we state the basic definitions and properties of the Fourier transform and convolution. Readers who are familiar with the topic might skip this section, except for \Def{Ft}.

\begin{definition}
The \emph{Hadamard gate} is \smash{$H \defeq \frac{1}{\sqrt{2}} \bigl( \begin{smallmatrix*}[r] 1 & 1 \\ 1 & -1 \end{smallmatrix*} \bigr)$}.
\end{definition}

\begin{definition}\label{def:Fourier}
The \emph{Fourier transform} of a function $F\colon \Z_2^n \to \R$ is a function $\hat{F}\colon \Z_2^n \to \R$ defined as $\hat{F}(w) \defeq \bra{w} H\xp{n} \ket{F}$ where $\ket{F} \defeq \sum_{x \in \Z_2^n} F(x) \ket{x}$. Here $\hat{F}(w)$ is called the \emph{Fourier coefficient} of $F$ at $w \in \Z_2^n$. Explicitly, $\hat{F}(w) = \frac{1}{\sqrt{2^n}} \sum_{x \in \Z_2^n} (-1)^{w \cdot x} F(x)$ where $x \cdot y \defeq \sum_{i=1}^n x_i y_i$. The set $\set{\hat{F}(w) \colon w \in \Z_2^n}$ is called the \emph{Fourier spectrum} of $F$.
\end{definition}

To define the Fourier transform of a Boolean function $f\colon \Z_2^n \to \Z_2$, we identify $f$ with a real-valued function $F\colon \Z_2^n \to \R$ in a canonical way: $F(x) \defeq (-1)^{f(x)} / \sqrt{2^n}$. Note that $F$ is normalized: $\sum_{x \in \Z_2^n} \abs{F(x)}^2 = 1$. Now we can abuse \Def{Fourier} as follows:
\begin{definition}\label{def:FourierZd}
The \emph{Fourier transform} of $f\colon \Z_2^n \to \Z_2$ is $\hat{F}(w) = \frac{1}{2^n} \sum_{x \in \Z_2^n} (-1)^{w \cdot x + f(x)}$.
\end{definition}
To avoid confusion, we use lower case letters for $\Z_2$-valued functions and capital letters for $\R$-valued functions.

\begin{definition}\label{def:Convolution}
The \emph{convolution} of functions $F, G\colon \Z_2^n \to \R$ is a function $(F * G)\colon \Z_2^n \to \R$ defined as $(F * G)(x) \defeq \sum_{y \in \Z_2^n} F(y) G(x - y)$. The \emph{$t$-fold convolution} of $F\colon \Z_2^n \to \R$ is a function $F^{*t}\colon \Z_2^n \to \R$ defined as
\begin{equation}
  F^{*t}(w)
  \defeq (\underbrace{F * \dotsb *F}_t)(w)
\, = \!\!\!\!\!\!\!\!
    \sum_{y_1, \dotsc, y_{t-1} \in \Z_2^n} \!\!\!\!\!
    F(y_1) \dotsb F(y_{t-1}) F \bigl( w - (y_1 + \dotsb + y_{t-1}) \bigr).
\end{equation}
\end{definition}

\begin{fact*}
Let $F,G,H\colon \Z_2^n \to \R$ denote arbitrary functions. The Fourier transform and convolution have the following basic properties:
\begin{enumerate}
  \item The Fourier transform is linear: $\widehat{F+G\:} = \hat{F} + \hat{G}$.
  \item The Fourier transform is self-inverse: $\hat{\hat{F}} = F$.
  \item Since $H\xp{n}$ is unitary, the Plancherel identity $\sum_{w \in \Z_2^n} \abs{\hat{F}(w)}^2 = \sum_{x \in \Z_2^n} \abs{F(x)}^2$ holds.
  \item Convolution is commutative ($F * G = G * F$) and associative ($(F * G) * H = F * (G * H)$).
  \item The Fourier transform and convolution are related through the following identities: $(\hat{F} * \hat{G}) / \sqrt{2^n} = \widehat{FG}$ and $(\widehat{F * G}) / \sqrt{2^n} = \hat{F} \hat{G}$, where $FG\colon \Z_2^n \to \C$ is the entry-wise product of functions $F$ and $G$: $(FG)(x) \defeq F(x) G(x)$.
  \item By induction, the $t$-fold convolution satisfies the identity $\bigl[ \hat{F} / \sqrt{2^n} \bigr]^{*t} = \widehat{F^t} / \sqrt{2^n}$.
\end{enumerate}
\end{fact*}

The following $t$-fold generalization of the Fourier spectrum plays a key role:
\begin{definition}\label{def:Ft}
For $t \geq 1$, the \emph{$t$-fold Fourier coefficient} of $f\colon \Z_2^n \to \Z_2$ at $w \in \Z_2^n$ is $\FC{t}(w) \defeq \sqrt{ \bigl[ \hat{F}^2 \bigr]^{*t} (w) }$. In particular, for $t = 1$ we have $\FC{1}(w) = \abs{\hat{F}(w)}$.
\end{definition}
We can express $\FC{t}(w)$ in many equivalent ways using the identities listed above:
\begin{equation}
  \bigl[ \FC{t}(w) \bigr]^2
= \bigl[ \hat{F}^2 \bigr]^{*t} (w)
= \biggl[ \frac{1}{\sqrt{2^n}} \bigl(\widehat{F*F}\bigr) \biggr]^{*t} \!\!\!\! (w)
= \frac{1}{\sqrt{2^n}} \widehat{\:(F*F)^t\:} (w).
\label{eq:FS norm}
\end{equation}

%%%%%%%%%%%%%%%%%%%%%%%%%%%%%%%%%%%%%%%%%%%
\section{Characterization of extreme cases} \label{sect:Extremes}
%%%%%%%%%%%%%%%%%%%%%%%%%%%%%%%%%%%%%%%%%%%

In this section we explore the set of functions for which the quantum query complexity of the $\BHSP$ is extreme. Recall that the $\BHSP$ can be solved with one query for bent functions and with $\Theta(\sqrt{2^n})$ queries for delta functions. Here we prove that $\BHSP_f$ can be solved exactly with one query only if $f$ is bent, and with $O(\sqrt{2^n})$ queries (with bounded error) for any $f$.

%----------------------------------%
\subsection{Easy functions are bent} \label{sect:Easy = bent}
%----------------------------------%

In general, the quantum query complexity of the $\BHSP$ for an arbitrary function is unknown. However, the problem becomes particularly easy for \emph{bent functions}, where a single query suffices to solve the problem exactly~\cite{Roetteler:2010}. In fact, bent functions are the only functions with this property, as we show here.

Bent functions can be characterized in many equivalent ways~\cite{BentFunctions, Dillon72}. The standard definition is that bent functions have a ``flat'' Fourier spectrum:
\begin{definition}\label{def:Bent}
A Boolean function $f\colon \Z_2^n \to \Z_2$ is \emph{bent} if all its Fourier coefficients $\hat{F}(w)$ (see \Def{FourierZd}) have the same absolute value: $\abs{\hat{F}(w)} = 1/\sqrt{2^n}$ for all $w \in \Z_2^n$.
\end{definition}

While many examples of bent functions have been constructed (e.g., see~\cite{MS:77, Dillon:75, Dobbertin:95}), no complete classification is known. As an example, the \emph{inner product} of two $n$-bit strings (modulo two) is a bent function~\cite{Dillon72,MS:77}: $\mathrm{IP}_n(x_1, \dotsc, x_n, y_1, \dotsc, y_n) \defeq \sum_{i=1}^n x_i y_i$.

We make a few simple observations about bent functions. Recall from \Sect{QFT and convolution} that the Fourier spectrum of $f$ is normalized as \smash{$\sum_{w \in \Z_2^n} \abs{\hat{F}(w)}^2 = 1$}, so the spectrum is ``flat'' only when \smash{$\abs{\hat{F}(w)} = 1/\sqrt{2^n}$} for all $w \in \Z_2^n$. Recall from \Def{FourierZd} that $\hat{F}(w)$ is always an integer multiple of $1/2^n$. Thus an $n$-variable function can only be bent if $n$ is even~\cite{Dillon:75, MS:77}. Moreover, from $\abs{\hat{F}(0)} = 1/\sqrt{2^n}$ we get that $\abs{\sum_{w \in \Z_2^n} (-1)^{f(x)}} = \sqrt{2^n}$, so a bent function $f$ is close to being balanced: $\abs{f} = (2^n \pm \sqrt{2^n}) / 2$ where $\abs{f} \defeq \abs{\set{x \in \Z_2^n \colon f(x) = 1}}$ is the \emph{Hamming weight} of $f$.

Our main result regarding bent functions is as follows.

\begin{restatable}{theorem}{BENT}\label{thm:Bent exact}
Let $f\colon \Z_2^n \to \Z_2$ be a Boolean function with $n \geq 2$. A quantum algorithm can solve $\BHSP_f$ exactly with a single query to $O_{f_s}$ if and only if $f$ is bent.
\end{restatable}

The proof is based on a characterization of an exact one-query quantum algorithm using a system of linear equations.  This system can be analyzed in terms of the autocorrelation of $f$, which in turn characterizes whether $f$ is bent.  The proof appears in \Apx{Bent}.

%------------------------------------%
\subsection{Hard functions} \label{sect:Hard = peaked}
%------------------------------------%

In this section we study hard instances of the $\BHSP$. First, we observe that the quantum query complexity of solving $\BHSP_f$ for any function $f$ is $O(\sqrt{2^n})$.

\begin{theorem}
\label{thm:allfunctions}
For any $f\colon \Z_2^n \to \Z_2$, the bounded-error quantum query complexity of $\BHSP_f$ is  $O(\sqrt{2^n})$.
\end{theorem}

If we view $f$ as a $2^n$-bit string indexed by $x \in  \Z_2^n$, this is a special case of the oracle identification problem considered by Ambainis et al.~\cite[Theorem 3]{AIK+04}, who show the following.

\begin{theorem}[Oracle Identification Problem]
\label{thm:OIP}
Given oracle access to an unknown $N$-bit string with the promise that it is one of $N$ known strings, the bounded-error quantum query complexity of identifying the unknown string is $O(\sqrt{N})$.
\end{theorem}

In the $\BHSP$, we have $N \defeq 2^n$. By \Thm{allfunctions}, the hardest functions are those with query complexity $\Omega(\sqrt{N})$. We know that delta functions have this query complexity, but are there any other functions that are as hard? The delta functions have $\abs{f} = 1$ (recall that $\abs{f}$ denotes the Hamming weight of $f$). Next we show that as $\abs{f}$ increases, the query complexity strictly decreases at first, until $\abs{f}=\Theta(\sqrt{N})$. For example, functions with $\abs{f} = 2$ have strictly smaller query complexity than the delta functions. However, as we approach $\abs{f} = \Omega(N)$, our upper bound is $\Theta(\sqrt{N})$ again. Without loss of generality, we assume that $\abs{f} \leq N/2$; otherwise we can simply negate the function to obtain a function with $\abs{f} \leq N/2$ that has exactly the same query complexity. Formally, we show the following refinement of \Thm{allfunctions}.

\begin{theorem}
\label{thm:allrefined}
For any $f\colon \Z_2^n \to \Z_2$ with $1 \leq \abs{f} \leq N/2$, the bounded-error quantum query complexity of $\BHSP_f$ is at most $\frac{\pi}{4} \sqrt{{N}/{\abs{f}}} + O(\sqrt{\abs{f}})$.
\end{theorem}

\begin{proof}
The algorithm has two parts. First we look for a ``1'' in the bit string contained in the oracle, i.e., an $x$ such that $f(x)=1$. This can be done by a variant of Grover's algorithm that finds a ``1'' in a string of length $N$ using at most $\frac{\pi}{4}\sqrt{N / \abs{f}}$ queries \cite{BoyerBHT98}. Now we have an $x$ such that $f_s(x)=1$ for some unknown $s$. Note that there can be at most $\abs{f}$ shifts $s$ with this property, because each corresponds to a distinct solution to $f(x+s) = 1$ and there are only $\abs{f}$ solutions to this equation.

We are now left with $\abs{f}$ candidates for the black-box function. Viewing this as an oracle identification problem, we have oracle access to an $N$-bit string that could be one of $\abs{f}$ possible candidates. Although the string has length $N$, there are only $\abs{f}$ potential candidates, so intuitively it seems like we should be able restrict the strings to length $\abs{f}$ and apply \Thm{OIP} to obtain the desired result. 

Formally, it can be shown that given $k \geq 2$ distinct Boolean strings of length $N$, there is a subset of indices, $S$, of size at most $k-1$, such that all the strings are distinct when restricted to $S$. We show this by induction. The base case is easy: we can choose any index that differentiates the two distinct strings. Now say we have $m$ distinct strings $y_1, y_2, \ldots, y_m$ and a subset of indices $S$ of size at most $m-1$, such that the $m$ strings are distinct on $S$. We want to add another string $y_{m+1}$ and increase the size of $S$ by at most 1. If $y_{m+1}$ differs with $y_1, y_2, \ldots, y_m$ on $S$, then we do not need to add any more indices to $S$ and we are done. If $y_{m+1}$ agrees with one of $y_1, y_2, \ldots, y_m$ on all of $S$, first note that it can only agree with one such string; to differentiate between these two, we add any index at which they differ to $S$, which must exist since they are distinct.
\end{proof}

This shows that a function can be hard---i.e., can have query complexity $\Theta(\sqrt{N})$---only if $\abs{f}$ is $O(1)$ or $\Theta(N)$.

Note that there do exist hard functions with $\abs{f} = \Theta(N)$. For example, consider the following function: $f(x) = 1$ if the first bit of $x$ is 1 or if $x$ is the all-zero string. This essentially embeds a delta function on the last $n-1$ bits, and thus requires $\Theta(\sqrt{N})$ queries. This function has  $\abs{f}=N/2 + 1$. However, there are also easy functions with $\abs{f} = \Theta(N)$, namely the bent functions.  Thus the Hamming weight does not completely characterize the hardness of the $\BHSP$ at high Hamming weight. However, it precisely characterizes the quantum query complexity at low Hamming weight:

\begin{theorem}
\label{thm:alllowerbound}
For any $f\colon \Z_2^n \to \Z_2$ with no undetectable shifts, the bounded-error quantum query complexity of $\BHSP_f$ is $\Omega(\sqrt{{N}/{\abs{f}}})$.
\end{theorem}

This follows from a simple application of the quantum adversary argument, with the adversary matrix taken to be the all ones matrix with zeroes on the diagonal. It also follows from Theorem 4 of \cite{AIK+04}.

%%%%%%%%%%%%%%%%%%%%%%%%%%
\section{The PGM approach} \label{sect:PGM}
%%%%%%%%%%%%%%%%%%%%%%%%%%

We now present an approach to the Boolean hidden shift problem based on the pretty good measurement (PGM)~\cite{PGM}. In particular, this approach shows that the Boolean hidden shift problem for random functions has small query complexity (see \Sect{Random}).

The main idea of the PGM approach is as follows. We apply the oracle on the uniform superposition and prepare $t$ independent copies of the resulting state (see \Sect{Phi}). Then we use knowledge of the function $f$ to perform the PGM in order to extract the hidden shift $s$ (see \Sect{PGM definition}). A similar strategy was used to efficiently solve the hidden subgroup problem for certain semidirect product groups, including the Heisenberg group~\cite{BCvD05},
and was subsequently applied to a hidden polynomial problem~\cite{Wocjan09}.

%-----------------------------------------------------------------%
\subsection{Performing \texorpdfstring{$t$}{t} queries in parallel} \label{sect:Phi}
%-----------------------------------------------------------------%

In this section we describe a quantum circuit that prepares a state with $w \cdot s$ encoded in the phase, where $s$ is the hidden shift and $w$ is the label of the corresponding standard basis vector. We use this circuit $t$ times in parallel, followed by a sequence of CNOTs, to prepare a certain state $\ket{\Phi^t(s)}$. In the next section we perform a PGM on these states for different values of $s$.

%.....................%
\subsubsection{Circuit}
%.....................%

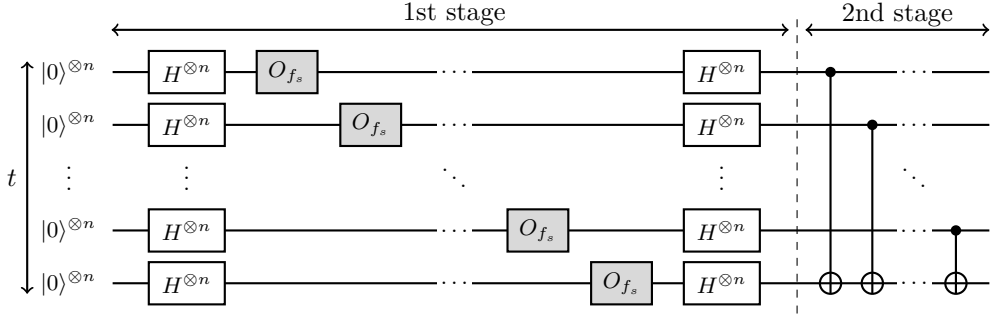
\begin{figure}
  \centering
  % !TeX root = ../BHSP.tex

\def\h{0.70cm}
\def\w{1.10cm}

\newcommand{\medsize}[1]{\normalsize{#1}}

\begin{tikzpicture}
  [font = \footnotesize,
   dot/.style = {circle, draw = black, fill = black, inner sep = 0mm, minimum size = 1.1mm},
   crc/.style = {circle, draw = black, fill = white, inner sep = 0mm, minimum size = 1.1mm},
   gate/.style = {draw = black, fill = white,    thick, minimum width = 1.00cm, rectangle},
   orac/.style = {draw = black, fill = black!15, thick, minimum width = 0.80cm, rectangle}
  ]

  % W I R E S

  \def\I{ 0.0*\w}
  \def\W{10.6*\w}

  \newcommand{\wire}[2]{
    \draw (\I,#2*\h) node (#1) {};
    \draw [-, thick] (#1) to (\W,#2*\h);
  }

  \foreach \i/\y in {q1/2, q2/1, q3/-1, q4/-2} {
    \wire{\i}{\y}
  }

  %  A R R O W S

  % Vertical
  \draw [dashed] ( 8.3*\w,3.0*\h) to ( 8.3*\w,-2.6*\h);
  \draw [<->,thick] (-1.0,2.2*\h) to (-1.0,-2.2*\h);
  \draw (-1.2,0) node {\medsize{$t$}};
  % Horizontal
  \draw [<->,thick] (0.1*\w,2.8*\h) to (8.2*\w,2.8*\h);
  \draw [<->,thick] (8.4*\w,2.8*\h) to (\W ,2.8*\h);
  \draw (4.2*\w,3.1*\h) node {\medsize{1st stage}};
  \draw (9.5*\w,3.1*\h) node {\medsize{2nd stage}};

  %  S T A T E S

  \foreach \i in {q1, q2, q3, q4} {
    \draw (\i) +(-0.95,0) node [right] {$\ket{0}\xp{n}$};
  }

  % G A T E S

  \newcommand{\gate}[5]{
    \draw (#1) + (#2*\w,0) node (#4) [#3, minimum height = 0.8*\h] {#5};
  }

  % Fourier transform
  \foreach \i in {1,2,3,4} {
    \gate{q\i}{1.0}{gate}{QFT\i}{$H\xp{n}$};
    \gate{q\i}{7.4}{gate}{QFT\i}{$H\xp{n}$};
  }
  \draw (-0.45 ,0.1) node {$\vdots$};
  \draw (1.0*\w,0.1) node {$\vdots$};
  \draw (7.4*\w,0.1) node {$\vdots$};

  % Oracle calls
  \foreach \i in {1,2,3,4} {
    \pgfmathparse{\i+(\i>=3)+1.2};
    \let\j=\pgfmathresult;
    \gate{q\i}{\j}{orac}{O\i}{$O_{f_s}$};
  }

  % Dots
  \foreach \i in {1,2,3,4} {
    \draw (q\i) +(4.2*\w,0) node [fill = white, inner sep = 0.5mm] {$\ldots$};
    \draw (q\i) +(9.7*\w,0) node [fill = white, inner sep = 0.5mm] {$\ldots$};
  }
  \draw (4.2*\w,0.1) node {$\ddots$};
  \draw (9.7*\w,0.1) node {$\ddots$};

  % CNOTs
  \def\r{0.2*\h}
  \foreach \i in {1,2,3} {
    \pgfmathparse{(\i+(\i>=3))/2+8.2};
    \let\j=\pgfmathresult;
%    \draw (q1)  +(\j*\w,0) circle [radius = \r] [thick];
%    \draw (q\i) +(\j*\w,0) node (D\i) [dot] {} to (\j*\w,2*\h+\r) [thick];
    \draw (q4)  +(\j*\w,0) circle [radius = \r] [thick];
    \draw (q\i) +(\j*\w,0) node (D\i) [dot] {} to (\j*\w,-2*\h-\r) [thick];
  }
\end{tikzpicture}
  \caption[Quantum algorithm for preparing the $t$-fold Fourier sate $\ket{\Phi^t(s)}$]{Quantum algorithm for preparing the $t$-fold Fourier sate $\ket{\Phi^t(s)}$ in \Eq{Phi}. The state on any register at the end of the first stage is given in \Eq{phase-single}.}
  \label{fig:phase}
\end{figure}

The circuit for preparing $\ket{\Phi^t(s)}$ appears in \Fig{phase}. It consists of two stages. The first stage prepares $t$ identical copies of the same state by using one oracle call between two quantum Fourier transforms on each register independently. Recall from \Def{Oracle} that the oracle acts on $n$ qubits and encodes the function in the phase: $O_{f_s}\colon \ket{x} \mapsto (-1)^{f(x+s)} \ket{x}$. The second stage entangles the states by applying a sequence of transversal controlled-NOT gates acting as $\ket{x} \ket{y} \mapsto \ket{x} \ket{y+x}$ for $x,y \in \Z_2^n$.

Note that all unitary post-processing after the oracle queries can be omitted since it does not affect the distinguishability of the states. We include it only to simplify the analysis.

%......................%
\subsubsection{Analysis} \label{sect:Analysis}
%......................%

During the first stage of the circuit, the first register evolves under $H\xp{n} \, O_{f_s} H\xp{n}$ (see \Fig{phase}):
\begin{equation}
  \ket{0}\xp{n}
  \mapsto \frac{1}{\sqrt{2^n}} \sum_{x \in \Z_2^n} \ket{x}
  \mapsto \frac{1}{\sqrt{2^n}} \sum_{x \in \Z_2^n} (-1)^{f(x + s)} \ket{x}
  \mapsto \frac{1}{2^n} \sum_{x,y \in \Z_2^n} (-1)^{f(x + s) + x \cdot y} \ket{y}.
\end{equation}
We can rewrite the resulting state as follows:
\begin{equation}
    \sum_{y \in \Z_2^n} (-1)^{s \cdot y}
    \Biggl( \frac{1}{2^n} \sum_{x \in \Z_2^n} (-1)^{f(x) + x \cdot y} \Biggr)
    \ket{y}
  = \sum_{y \in \Z_2^n} (-1)^{s \cdot y} \hat{F}(y) \ket{y}.
  \label{eq:phase-single}
\end{equation}
The overall state after the first stage is just the $t$-fold tensor product of the above state:
\begin{equation}
  \sum_{y_1, \dotsc, y_t \in \Z_2^n}
  (-1)^{s \cdot (y_1 + \dotsb + y_t)}
  \bigotimes_{i=1}^t \hat{F}(y_i) \ket{y_i}.
\end{equation}

In the second stage of the algorithm, the controlled-NOT gates transform this state into
\begin{align}
 & \sum_{y_1, \dotsc, y_t \in \Z_2^n}
   (-1)^{s \cdot (y_1 + \dotsb + y_t)}
   \Biggl[ \bigotimes_{i=1}^{t-1} \hat{F}(y_i) \ket{y_i} \Biggr]
   \hat{F}(y_t) \ket{y_1 + \dotsb + y_t} \\
=& \sum_{y_1, \dotsc, y_t \in \Z_2^n}
   (-1)^{s \cdot y_t}
   \Biggl[ \bigotimes_{i=1}^{t-1} \hat{F}(y_i) \ket{y_i} \Biggr]
   \hat{F} \bigl( y_t - (y_1 + \dotsb + y_{t-1}) \bigr) \ket{y_t}.
\end{align}
We can rewrite this state as
\begin{equation}
  \ket{\Phi^t(s)} \defeq \sum_{w \in \Z_2^n} (-1)^{s \cdot w} \ket{\FS{t}{w}} \ket{w},
  \label{eq:Phi}
\end{equation}
where the non-normalized state $\ket{\FS{t}{w}}$ on $(t-1) n$ qubits is given by
\begin{equation}
  \ket{\FS{t}{w}} \, \defeq \!\!\!\!\!\!\!\!
  \sum_{y_1, \dotsc, y_{t-1} \in \Z_2^n} \!\!\!\!\!\!
  \hat{F}(y_1) \dotsb \hat{F}(y_{t-1}) \hat{F} \bigl( w - (y_1 + \dotsb + y_{t-1}) \bigr)
  \ket{y_1} \dotsb \ket{y_{t-1}}.
  \label{eq:t-Fourier state}
\end{equation}
Its norm is just the $t$-fold Fourier coefficient: $\norm{\ket{\FS{t}{w}}} = \FC{t}(w)$ (see \Def{Ft}).

%--------------------------------------%
\subsection{The pretty good measurement} \label{sect:PGM definition}
%--------------------------------------%

Let $\set{\rho_s^{(t)} \colon s \in \Z_2^n}$ be a set of mixed states where $\rho_s^{(t)}$ is given with probability $p_s$. The \emph{pretty good measurement} (PGM)~\cite{PGM} for discriminating these states is a POVM with operators $\set{E_s \colon s \in \Z_2^n} \cup \set{E_{*}}$ where
\begin{align}
  E_s    &\defeq E^{-1/2} \, p_s \rho_s^{(t)} \, E^{-1/2}, &
  E      &\defeq \sum_{s \in \Z_2^n} p_s \rho_s^{(t)}, &
  E_{*}  &\defeq I - \sum_{s \in \Z_2^n} E_s.
\end{align}
In our case, $\rho_s^{(t)} \defeq \ket{\Phi^t(s)} \bra{\Phi^t(s)}$ and $p_s \defeq 1/2^n$ where $\ket{\Phi^t(s)}$ is defined in \Eq{Phi}.

To find the operators $E_s$, we compute
\begin{align}
  E &= \sum_{s \in \Z_2^n} \frac{1}{2^n} \sum_{w,w' \in \Z_2^n}
       (-1)^{(w + w') \cdot s} \ket{\FS{t}{w}} \bra{\FS{t}{w'}} \x \ket{w} \bra{w'} \\
    &= \sum_{w \in \Z_2^n} \norm{\ket{\FS{t}{w}}}^2 \cdot
       \frac{\ket{\FS{t}{w}} \bra{\FS{t}{w}}}{\norm{\ket{\FS{t}{w}}}^2} \x \ket{w} \bra{w}.
\end{align}
From now on we use the convention that terms with $\norm{\ket{\FS{t}{w}}} = 0$ are omitted from all sums. As $E$ is a sum of mutually orthogonal rank-$1$ operators with eigenvalues $\norm{\ket{\FS{t}{w}}}^2$, we find
\begin{equation}
  E^{-1/2} = \sum_{w \in \Z_2^n} \frac{1}{\norm{\ket{\FS{t}{w}}}} \cdot
             \frac{\ket{\FS{t}{w}} \bra{\FS{t}{w}}}{\norm{\ket{\FS{t}{w}}}^2} \x \ket{w} \bra{w}.
\end{equation}
Note that $E_s = \proj{E_s}$ where $\ket{E_s} \defeq E^{-1/2} \sqrt{p_s} \ket{\Phi^t(s)}$. We can express $\ket{E_s}$ as follows:
\begin{align}
  \ket{E_s}
 &= \Biggl(
      \sum_{w \in \Z_2^n} \frac{\ket{\FS{t}{w}} \bra{\FS{t}{w}}}{\norm{\ket{\FS{t}{w}}}^3} \x
      \ket{w} \bra{w}
    \Biggr) \frac{1}{\sqrt{2^n}}
    \Biggl(
      \sum_{w \in \Z_2^n} (-1)^{w \cdot s} \ket{\FS{t}{w}} \ket{w}
    \Biggr) \\
 &= \frac{1}{\sqrt{2^n}} \sum_{w \in \Z_2^n} (-1)^{w \cdot s}
    \frac{\ket{\FS{t}{w}}}{\norm{\ket{\FS{t}{w}}}} \x
    \ket{w}. \label{eq:Es}
\end{align}
Notice that the vectors $\ket{E_s}$ are orthonormal, so the PGM is just an orthogonal measurement in this basis (with another outcome corresponding to the orthogonal complement).  Therefore the measurement is unambiguous: if it outputs a value of $s$ (rather than the inconclusive outcome $*$) then it is definitely correct. The corresponding zero-error algorithm can be summarized as follows:

\begin{algobox}{0.65}
\PGMtxt$(f,t)$\label{alg:PGM}
\begin{enumerate}
 \item Prepare $\ket{\Phi^t(s)}$ using the circuit shown in \Fig{phase}.
 \item Recover $s$ by performing an orthogonal measurement with projectors $\set{\proj{E_s} \colon s \in \Z_2^n} \cup \set{E_{*}}$.
\end{enumerate}
\end{algobox}
\vspace{1pt}

\begin{lemma}\label{lem:PGM}
The $t$-query algorithm \PGM$(f,t)$ solves $\BHSP_f$ with success probability
\begin{equation}
  p_f(t) \defeq \Biggl( \frac{1}{\sqrt{2^n}} \sum_{w \in \Z_2^n} \FC{t}(w) \Biggr)^2,
  \label{eq:avg t-fold coefficient}
\end{equation}
where $\FC{t}(w) = \norm{\ket{\FS{t}{w}}}$ denotes the $t$-fold Fourier spectrum of $f\colon \Z_2^n \to \Z_2$ (see \Def{Ft}).
\end{lemma}

\begin{proof}
Recall that the PGM for discriminating the states $\ket{\Phi^t(s)} = \sum_{w \in \Z_2^n} (-1)^{s \cdot w} \ket{\FS{t}{w}} \ket{w}$ from \Eq{Phi} is an orthogonal measurement on $\ket{E_s}$ (defined in \Eq{Es}) and the orthogonal complement. Thus, given the state $\ket{\Phi^t(s)}$, the success probability to recover the hidden shift $s$ correctly is $\abs[\big]{\braket{E_s}{\Phi^t(s)}}^2$. This is equal to the expression in \Eq{avg t-fold coefficient}. Moreover, it does not depend on $s$, so $p_f(t)$ is the success probability even if $s$ is chosen adversarially as in the definition of $\BHSP_f$ (\Prob{ShiftZd}). Note that the convention of omitting terms with $\norm{\ket{\FS{t}{w}}} = 0$ is consistent since such terms do not appear in \Eq{avg t-fold coefficient}.
\end{proof}

We can use \Eq{FS norm} to write the success probability as
\begin{equation}
  p_f(t)
  = \frac{1}{2^n}
      \Biggl(
        \sum_{w \in \Z_2^n} \sqrt{\frac{1}{\sqrt{2^n}} \widehat{\:(F*F)^t\:}(w)}
      \Biggr)^2.
    \label{eq:success probability}
\end{equation}
Recall from \Sect{QFT and convolution} that $\FC{1}(w) = \abs{\hat{F}(w)}$, so for $t = 1$ we have
\begin{equation}
  p_f(1)
  = \frac{1}{2^n} \Biggl( \sum_{w \in \Z_2^n} \abs{\hat{F}(w)} \Biggr)^2.
  \label{eq:PGM success probability with t=1}
\end{equation}

%-------------------------------%
\subsection{Performance analysis} \label{sect:PGM performance}
%-------------------------------%

In this section we analyze the performance of the PGM algorithm described above on several different classes of Boolean functions. For delta functions our algorithm performs worse than Grover's algorithm. On the other hand, for bent and random functions it needs only one and two queries, respectively.

%.............................%
\subsubsection{Delta functions} \label{sect:PGM delta}
%.............................%

Let us check how our algorithm performs when $f$ is a delta function, i.e., $f(x) = \delta_{x,x_0}$ for some $x_0 \in \Z_2^n$. 
A simple calculation using the Fourier spectrum of a delta function shows that the success probability of \PGM$(f,t)$ is
\begin{equation}
  p_{f}(t) =
  \frac{1}{2^{2n}} \left(
    \left(2^n-1\right) \sqrt{1-\left(\frac{2^n-4}{2^n}\right)^t}
  + \sqrt{1+\left(2^n-1\right) \left(\frac{2^n-4}{2^n}\right)^t}
  \right)^2.
\end{equation}
Unfortunately, if we choose $t = \sqrt{2^n}$, then the success probability goes to $0$ as $n \rightarrow \infty$. In fact, the same happens even if $t = c^n$ for any $c < 2$. Only if we take $t = 2^n$ does the success probability approach a positive constant $1 - 1/e^4 \approx 0.98$ as $n \rightarrow \infty$. This means that the PGM algorithm does not give us the quadratic speedup of Grover's algorithm. (Indeed, this follows from the more general fact that quantum speedup for unstructured search cannot be parallelized~\cite{Zalka99}.) Thus the PGM algorithm is not optimal in general.

%............................%
\subsubsection{Bent functions}
%............................%

Let $f$ be a Bent function. Recall from \Sect{Easy = bent} that its Fourier spectrum is ``flat'', i.e., $\abs{\hat{F}(w)} = 1/\sqrt{2^n}$ for all $w \in \Z_2^n$. In this case, \Eq{PGM success probability with t=1} gives $p_{f}(1) = 1$, so we can find the hidden shift with certainty by measuring $\ket{\Phi^1(s)}$ with the pretty good measurement (recall that preparing $\ket{\Phi^1(s)}$ requires only one query to $O_{f_{s}}$), reproducing a result of R{\"o}tteler.

\begin{theorem}[\cite{Roetteler:2010}]\label{thm:Bent PGM}
If $f$ is a bent function then a quantum algorithm can solve $\BHSP_f$ exactly using a single query to $O_{f_{s}}$.
\end{theorem}

%..............................%
\subsubsection{Random functions} \label{sect:Random}
%..............................%

For random Boolean functions, our algorithm performs almost as well as for bent functions. For random $f$, we are only able to show that the expected success probability of the one-query algorithm \PGM$(f,1)$ is at least $2/\pi + o(1)$ for large $n$ (see \Thm{Random1} in \Apx{Random1}), so the algorithm only succeeds with constant probability, which cannot easily be boosted. However, the expected success probability of the two-query algorithm \PGM$(f,2)$ is exponentially close to $1$.

\begin{restatable}{theorem}{RANDOM}\label{thm:Random2}
Let $f$ be an $n$-argument Boolean function chosen uniformly at random and suppose that a hidden shift for $f$ is chosen adversarially. Then \PGM$(f,2)$ solves $\BHSP_f$ with expected success probability $\bar{p} \geq 1 - \frac{3}{64} \cdot 2^{-n}$.
\end{restatable}

The proof uses the second moment method to lower bound the expected success probability.  We compute the variance of the $2$-fold Fourier spectrum by relating it to the combinatorics of pairings.  The proof appears in \Apx{Random functions}.

\Thm{Random2} implies that our algorithm can determine the hidden shift with near certainty as $n \to \infty$. This is surprising since some functions, such as delta functions (see \Sect{Hard = peaked}), require $\Omega(\sqrt{2^n})$ queries. Furthermore, a randomly chosen function could have an undetectable shift (see \Sect{b-shifts}), in which case it is not possible in principle to completely determine an adversarially chosen shift with success probability more than $1/2$.

At first glance, \Thm{Random2} may appear to be a strengthening of the main result of \cite{GRR11}, which shows that $O(n)$ queries suffice to solve a version of the Boolean hidden shift problem for a random function.  However, while our approach uses dramatically fewer queries, the results are not directly comparable: Ref.~\cite{GRR11} considers a weaker model in which the unshifted function is given by an oracle rather than being known explicitly.  In particular, while the result of \cite{GRR11} gives an average-case exponential separation between classical and quantum query complexity, such a result is not possible in the model where the function is known explicitly. In this model, there cannot be a super-polynomial speedup for quantum computation.  This follows from general results from learning theory discussed at the end of \Sect{Intro}.  In particular, it follows that if the quantum query complexity of the problem for a function $f$ is $Q$, then the deterministic classical query complexity of the problem for the same function is at most $O(nQ^2)$ \cite{SG04}.

%%%%%%%%%%%%%%%%%%%%%%%%%%%%%%%%%%%%%%%%%%%%%%%%%%%%%%%%%%
\section{Quantum rejection sampling with parallel queries} \label{sect:QRS with t queries}
%%%%%%%%%%%%%%%%%%%%%%%%%%%%%%%%%%%%%%%%%%%%%%%%%%%%%%%%%%

In this section we explain a hybrid approach that combines the Quantum Rejection Sampling (QRS) algorithm for the $\BHSP$~\cite{QRS} with the PGM approach. The resulting algorithm does not require an extra amplification step for boosting the success probability, unlike the original QRS algorithm.

%-------------------------------------------------------%
\subsection{Original quantum rejection sampling approach} \label{sect:Basic QRS}
%-------------------------------------------------------%

\begin{theorem}[\cite{QRS}]\label{thm:QRS}
For a given Boolean function $f\colon \Z_2^n \to \Z_2$, define unit vectors $\vc{\quantump}, \vc{\quantums} \in \R^{2^n}$ as $\quantump_{\bv{w}} \defeq \abs{\hat{F}(\bv{w})}$ and $\quantums_{\bv{w}} \defeq 1/\sqrt{2^n}$ for $\bv{w} \in \Z_2^n$. Moreover, let
\begin{align}
  p_{\min} &\defeq (\vc{\quantums}\tp \cdot \vc{\quantump})^2
             = \frac{1}{2^n} \biggl( \sum_{\bv{w} \in \Z_2^n} \abs{\hat{F}(\bv{w})} \biggr)^2, &
  p_{\max} &\defeq \!\! \sum_{k\colon\quantump_k>0} \!\! \quantums_k^2
             = \frac{1}{2^n} \abs{\set{w \colon \hat{F}(w) \neq 0}}.
  \label{eq:ps}
\end{align}
For any desired success probability $p \in [p_{\min}, p_{\max}]$, the quantum rejection sampling algorithm solves $\BHSP_f$ with $O(1/\norm{\waterfilling})$ queries, where the ``water-filling'' vector $\waterfilling \in \R^{2^n}$ is defined in~\cite{QRS}.
\end{theorem}

In particular, if $p_{\max} = 1$ then the QRS algorithm can achieve any success probability arbitrarily close to $1$ with $O \bigl( 1/(\sqrt{2^n} \hat{F}_{\min}) \bigr)$ queries, where $\hat{F}_{\min} \defeq \min_w \abs{\hat{F}(w)}$. However, if $\hat{F}(w) = 0$ for some $w$, then from \Eq{ps} we see that $p_{\max} < 1$. In this case one needs an additional amplification step to boost the success probability (a method based on SWAP test was proposed in~\cite{QRS}). We show that this step can be avoided by using $t$ parallel queries in the original QRS algorithm for some $t \leq n$.

%------------------------------------------------------------------%
\subsection{Non-degenerate functions with almost vanishing spectrum}
%------------------------------------------------------------------%

Before explaining our hybrid approach, let us verify that there exist non-trivial functions with a large fraction of their Fourier spectrum equal to zero, so the issue discussed above applies.

It is easy to construct degenerate functions with the desired property. For example, if a function is shift-invariant, i.e., $f(x+s) = f(x)$ for some $s \in \Z_2^n$, then at least half of the Fourier spectrum of $f$ is guaranteed to be zero.  The same also happens if $f(x+s) = f(x) + 1$ (see \Lem{b-shifts} in \Sect{b-shifts}). However, such examples are not interesting, since a shift-invariant $n$-argument Boolean function is equivalent to an $(n-1)$-argument Boolean function (see \Sect{b-shifts} for more details).

Instead, we consider Boolean functions defined using decision trees. A \emph{decision tree} is a binary tree whose vertices are labeled by arguments of $f$ and whose leaves contain the values of $f$. An example of such tree and the rules for evaluating the corresponding function are given in \Fig{Tree}.

Without loss of generality, we can consider only decision trees where on each path from the root to a leaf no argument appears more than once (otherwise some parts of the tree would not be reachable). The length of a longest path from the root to a leaf is the \emph{height} of the tree. If a Boolean function is defined by a decision tree of height $h$, then all its Fourier coefficients with Hamming weight larger than $h$ are zero (see \Lem{Tree} in \Sect{Trees}). This observation can be used to construct non-degenerate Boolean functions with almost vanishing Fourier spectrum.

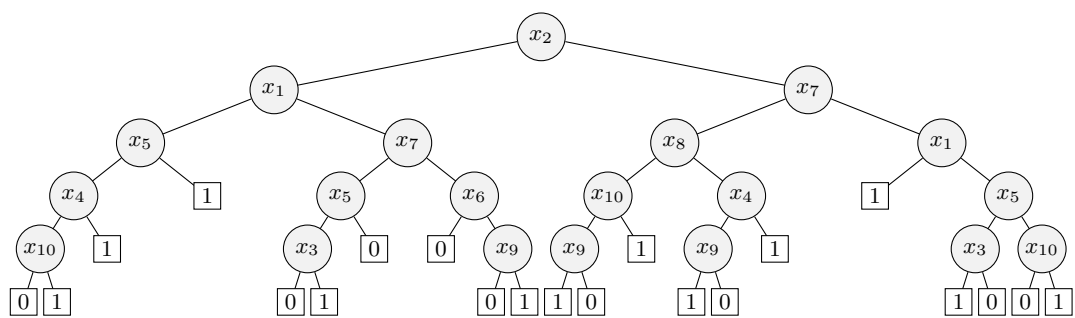
\begin{figure}
  \centering
  % !TeX root = ../BHSP.tex

%\documentclass[landscape]{article}
%\usepackage{tikz}
%\usetikzlibrary{arrows}
%\begin{document}

\begin{tikzpicture}
  [font = \footnotesize,
   vx/.style = {circle,    draw = black, fill = gray!10, minimum size = 18, inner sep = 0pt},
   lv/.style = {rectangle, draw = black, fill = white,   minimum size = 10, inner sep = 0pt},
%   lb/.style = {above, fill = white, font = \tiny, inner sep = 1pt}
  ]

\def\dx{12.5}
\def\dy{20}

\tikzstyle{level 1} = [level distance = \dy, sibling distance = 16*\dx]
\tikzstyle{level 2} = [level distance = \dy, sibling distance =  8*\dx]
\tikzstyle{level 3} = [level distance = \dy, sibling distance =  4*\dx]
\tikzstyle{level 4} = [level distance = \dy, sibling distance =  2*\dx]
\tikzstyle{level 5} = [level distance = \dy, sibling distance =  1*\dx]

\newcommand{\tree}[3]{
  node [vx] {$x_{#1}$}
  child {#2} %edge from parent node [lb] {$0$}}
  child {#3} %edge from parent node [lb] {$1$}}
}

\newcommand{\leave}[1]{
  node [lv] {$#1$}
}

\path
\tree{2}{
  \tree{1}{
    \tree{5}{
      \tree{4}{
        \tree{10}{\leave{0}}{\leave{1}} }{
        \leave{1}
      }
    }{
      \leave{1}
    }
  }{
    \tree{7}{
      \tree{5}{
        \tree{3}{\leave{0}}{\leave{1}} }{
        \leave{0}
      }
    }{
      \tree{6}{
        \leave{0} }{
        \tree{9}{\leave{0}}{\leave{1}}
      }
    }
  }
}{
  \tree{7}{
    \tree{8}{
      \tree{10}{
        \tree{9}{\leave{1}}{\leave{0}} }{
        \leave{1}
      }
    }{
      \tree{4}{
        \tree{9}{\leave{1}}{\leave{0}} }{
        \leave{1}
      }
    }
  }{
    \tree{1}{
      \leave{1}
    }{
      \tree{5}{
        \tree{3}{\leave{1}}{\leave{0}} }{
        \tree{10}{\leave{0}}{\leave{1}}
      }
    }
  }
};

\end{tikzpicture}

%\end{document}
  \caption[Decision tree for function $f_{10}$]{Decision tree for a $10$-argument Boolean function $f_{10}$. To compute the value of the function for given input $x_1, \dotsc, x_{10} \in \Z_2^n$, proceed down the tree starting from the root; move left if the corresponding argument is equal to $0$ or right if it is equal to $1$. Once a leaf is reached, its label is the value of the function for the given input. For example, $f_{10}(x_1, \dotsc, x_{10})$ evaluates to zero when $x_2 = x_1 = x_5 = x_4 = x_{10} = 0$, since the leftmost leaf has label zero. This tree has height five.}
  \label{fig:Tree}
\end{figure}

\begin{example*}
The $10$-argument Boolean function $f_{10}$ whose decision tree is shown in \Fig{Tree} has no shift invariance, yet $928$ (out of $2^{10} = 1024$) of its Fourier coefficients are zero.
\end{example*}

\newcommand\ttt{\texorpdfstring{$t$}{t}}
%-----------------------------------------------------------%
\subsection{The \ttt-fold Fourier spectrum as \ttt{} increases}
%-----------------------------------------------------------%

Let us now show how to deal with the zero Fourier coefficients. The main idea stems from the following observation: if $S_t \defeq \set{w \in \Z_2^n \colon \FC{t}(w) \neq 0}$ then $S_{t+1} = S_t + S_1$ (see \Prop{expansion} in \Sect{Increasing t}). If $S_1$ spans $\Z_2^n$, we can apply this recursively and eliminate all zeroes from the $t$-fold Fourier spectrum $\FC{t}$. In particular, it suffices to take $t \leq n$ (see \Lem{Eliminating zeroes} in \Sect{Increasing t}). For example, for $f_{10}$ the fraction of non-zero values of $\FC{t}$ for $t = 1, 2, 3, 4$ is $0.09$, $0.61$, $0.94$, $1$, respectively. In particular, $\FC{4}$ is non-zero everywhere.

%------------------------------------------------------------%
\subsection{Quantum rejection sampling with \ttt-fold queries} \label{sect:t-fold QRS}
%------------------------------------------------------------%

We can use quantum rejection sampling with $t$ queries in parallel to solve the $\BHSP$. Suppose we transform the $t$-fold Fourier state $\ket{\Phi^t(s)}$ from \Eq{Phi} into the PGM basis vector $\ket{E_s}$ defined in \Eq{Es} using QRS. This corresponds to setting $\quantump_{\bv{w}} = \FC{t}(w)$ and $\quantums_{\bv{w}} = 1/\sqrt{2^n}$. Since the circuit from \Fig{phase} can be used to prepare $\ket{\Phi^t(s)}$ with $t$ queries, \Thm{QRS} still holds if $\abs{\hat{F}(w)}$ is replaced by $\FC{t}(w)$ and the query complexity is multiplied by $t$. This observation together with \Lem{Eliminating zeroes} implies that as long as $f$ is not shift invariant, we can recover the hidden shift $s$ with success probability arbitrarily close to $1$ using quantum rejection sampling with some $t \leq n$.

\begin{theorem}\label{thm:QRSt}
Let $f\colon \Z_2^n \to \Z_2$ be a Boolean function and let $p$ be sufficiently large. Then $\BHSP_f$ can be solved with success probability $p$ using $O(t/\norm{\waterfilling})$ queries for some $t \in \set{1, \dotsc, n}$ where $\quantump_{\bv{w}} \defeq \FC{t}(w)$, $\quantums_{\bv{w}} \defeq 1/\sqrt{2^n}$, and the ``water-filling'' vector $\waterfilling \in \R^{2^n}$ is defined in~\cite{QRS}.
\end{theorem}

%%%%%%%%%%%%%%%%%%%%%
\section{Conclusions} \label{sect:Conclusions}
%%%%%%%%%%%%%%%%%%%%%

A comparison of quantum query complexity bounds for solving the $\BHSP$ for different classes of functions is given in \Tab{Complexity}. If the QRS algorithm works for random functions with $O(1)$ queries, then it is optimal up to constant factors in all three cases listed in the table. However, from \Sect{Basic QRS} we know that the basic QRS algorithm without amplification performs poorly when $f$ has many zero Fourier coefficients (which is the case, e.g., for the decision trees considered in \Sect{Trees}). This suggests that the basic (unamplified) QRS algorithm is likely not optimal in general.

\begin{table}[ht]
\begin{center}
\begin{tabular}{c|c|c|c|c}
\multicolumn{1}{c|}{\multirow{2}{*}{Approach}} &
\multicolumn{3}{c|}{Functions} & {\multirow{2}{*}{Comments}}
\\ \cline{2-4}
& delta & bent & random
\\ \hline
 PGM                     & $O(2^n)$             & $1$ & $2$ & zero error \\
 QRS~\cite{QRS}          & $O( \sqrt{2^n})$     & $1$ & ? \\
 ``Simon''~\cite{GRR11}  & $O(n\sqrt{2^n})$     & $O(n)$ & $O(n)$\phantom{?} & zero error, black-box $f$ okay \\
Learning theory \cite{AS05}  &$O(n\log n\sqrt{2^n})$ & $O(n\log n)$ & $O(n \log n)$ & optimal up to log factors $\forall\, f$
\\ \hline
Lower bounds: & $\Omega(\sqrt{2^n})$ & $1$ & $1$
\end{tabular}
\end{center}
\caption[Summary of quantum query complexity upper and lower bounds for $\BHSP$]{Summary of quantum query complexity upper and lower bounds for $\BHSP$. We do not know the query complexity of the QRS algorithm for random functions.}
\label{tab:Complexity}
\end{table}

The ``Simon''-type approach due to~\cite{GRR11} always has an overhead of a factor $O(n)$, reflecting the fact that at least $n$ linearly independent equations are needed to solve a linear system in $n$ variables. (Note that this approach works in the weaker model where the unshifted function is given by an oracle, so it still provides an upper bound when the function is known explicitly.) The learning theory approach \cite{AS05} also has logarithmic overhead. Finally, the PGM approach performs very well in the easy cases, the bent and random functions, but fails to provide any speedup for delta functions. As mentioned in \Sect{PGM delta}, this can be attributed to the fact that Grover's algorithm is intrinsically sequential.

In summary, none of the algorithms listed in \Tab{Complexity} is optimal. However, by combining these algorithms and possibly adding some new ideas, one might obtain an algorithm that is optimal for all Boolean functions. In particular, the QRS approach with $t$-fold queries appears promising.

We conclude by mentioning some open questions regarding the Boolean hidden shift problem:
\begin{enumerate}
  \item Find a query-optimal quantum algorithm for general functions (recall that the learning theory algorithm is only optimal up to logarithmic factors \cite{SG04,AS05}).
  \item Identify natural classes of Boolean functions lying between the two extreme cases of bent and delta functions (say, the decision trees considered in \Sect{Trees}) and characterize the quantum query complexity of the $\BHSP$ for these functions.
  \item Determine the number of queries required by the QRS algorithm for random functions.
  \item What is the query complexity of verifying a given shift? (A quantum procedure with one-sided error, based on the swap test, was given in~\cite{QRS}.)
  \item What is the quantum query complexity of extracting one bit of information about the hidden shift?
  \item What is the classical query complexity of the Boolean hidden shift problem?
  \item Can we say anything non-trivial about the time complexity of the Boolean hidden shift problem, either classically or quantumly?
  \item Can the $\BHSP$ for random functions be solved with a single query? Our approach based on the PGM only gives a lower bound on the expected success probability that approaches $2/\pi$ for large $n$ (see \Thm{Random1}), whereas we require a success probability that approaches $1$ as $n \to \infty$. It might be fruitful to consider querying the oracle with non-uniform amplitudes.
\end{enumerate}

Finally, it might be interesting to consider the generalization of the Boolean hidden shift problem to the case of functions $f\colon \Z_d^n \to \Z_d$.

%%%%%%%%%%%%%%%%%%%%%%%%%%%%%%%%
\subparagraph*{Acknowledgements}
%%%%%%%%%%%%%%%%%%%%%%%%%%%%%%%%

We thank J{\'e}r{\'e}mie Roland for useful discussions and Dmitry Gavinsky for suggesting to use decision trees to construct non-degenerate functions with many zero Fourier coefficients. Part of this work was done while AC and MO were visiting NEC Labs, and during the Quantum Cryptanalysis seminar (No.~11381) at Schloss Dagstuhl. This work was supported in part by NSERC, the Ontario Ministry of Research and Innovation, and the US ARO/DTO. MO acknowledges additional support from the DARPA QUEST program under contract number HR0011-09-C-0047.

  %%%%%
%%%%%%%%%
\appendix
%%%%%%%%%
  %%%%%

%%%%%%%%%%%%%%%%%%%%%%%%%%%%%%%%%%%%%
\section{Converse for bent functions} \label{apx:Bent}
%%%%%%%%%%%%%%%%%%%%%%%%%%%%%%%%%%%%%

The goal of this appendix is to prove \Thm{Bent exact}. First we need an alternative characterization of bent functions.

\begin{proposition}\label{prop:Bent condition}
A Boolean function $f$ is bent if and only if $(F * F)(x) = \delta_{x,0}$.
\end{proposition}

\begin{proof}
If $(F * F)(x) = \delta_{x,0}$, then using identities from \Sect{QFT and convolution}, we find
\begin{equation}
  \hat{F}^2(w)
  = \frac{1}{\sqrt{2^n}} (\widehat{F * F}) (w)
  = \frac{1}{2^n} \sum_{x \in \Z_2^n} (-1)^{w \cdot x} (F * F) (x)
  = \frac{1}{2^n}
\end{equation}
so $f$ is bent. Conversely, if $f$ is bent then
\begin{equation}
  (F * F)(w)
  = \sqrt{2^n} \widehat{\hat{F}^2}(w)
  = \sum_{x \in \Z_2^n} (-1)^{w \cdot x} \hat{F}^2(x)
  = \sum_{x \in \Z_2^n} (-1)^{w \cdot x} \frac{1}{2^n}
  = \delta_{w,0}
\end{equation}
and the result follows.
\end{proof}

\newcommand{\0}{\varnothing}

\BENT*

\begin{proof}
The most general one-query algorithm for solving $\BHSP_f$ using a controlled phase oracle (or equivalently, an oracle that computes the function in a register) performs a query on some superposition of all binary strings $x \in \Z_2^n$ and an extra symbol ``$\0$'' that allows for the possibility of not querying the oracle. Without loss of generality, the initial state is
\begin{equation}
  \alpha_{\0} \ket{\0} + \sum_{x \in \Z_2^n} \alpha_x \ket{x}
\end{equation}
for some amplitudes $\alpha_\0 \in \C$ and $\alpha_x \in \C$ for $x \in \Z_2^n$ such that $\abs{\alpha_{\0}}^2 + \sum_{x \in \Z_2^n} \abs{\alpha_x}^2 = 1$. The oracle acts trivially on $\ket{\0}$, so the state after the query is
\begin{equation}
  \ket{\phi_s} \defeq \alpha_{\0} \ket{\0}
  + \sum_{x \in \Z_2^n} \alpha_x (-1)^{f(x+s)} \ket{x}
\end{equation}
where $s \in \Z_2^n$ is the hidden shift. For an exact algorithm, we must have
\begin{equation}
  \forall\, s \neq s': \;
0 = \braket{\phi_s}{\phi_{s'}}
  = \abs{\alpha_{\0}}^2
  + \sum_{x \in \Z_2^n} \abs{\alpha_x}^2 (-1)^{f(x+s)+f(x+s')}.
  \label{eq:Exact}
\end{equation}

We can describe \Eq{Exact} as a linear system of equations. Define $p_{\0} \defeq \abs{\alpha_{\0}}^2$ and let $p$ be a sub-normalized probability distribution on $\Z_2^n$ defined by $p_x \defeq |\alpha_x|^2$. Let $M$ be a rectangular matrix with rows labeled by elements of $A \defeq \set{(s,s') \in \Z_2^n \times \Z_2^n : s \neq s'}$ and columns labeled by $x \in \Z_2^n$, with entries
\begin{equation}
  M_{ss'\!\!,x} \defeq (-1)^{f(x+s)+f(x+s')}.
\end{equation}
Then \Eq{Exact} is equivalent to
\begin{equation}
  M p = -p_{\0} u
  \label{eq:Exact M}
\end{equation}
where $u$ is the all-ones vector indexed by elements of $A$. In other words, there exists an exact one-query quantum algorithm for solving $\BHSP_f$ if and only if \Eq{Exact M} holds for some $p_{\0}$ and $p$ that together form a probability distribution on $\set{\0} \cup \Z_2^n$.

If $f$ is bent, there is an exact one-query quantum algorithm corresponding to $p_{\0} = 0$ and $p = \mu$, the uniform distribution (i.e., $\mu_x \defeq 1/2^n$ for all $x \in \Z_2^n$). Notice that the entries of the vector $M \mu$ are
\begin{align}
  (M\mu)_{ss'}
  &= \frac{1}{2^n} \sum_{x \in \Z_2^n} M_{ss'\!\!,x} \label{eq:Magic1} \\
  &= \frac{1}{2^n} \sum_{x \in \Z_2^n} (-1)^{f(x+s)+f(x+s')} \\
  &= \frac{1}{2^n} \sum_{x \in \Z_2^n} (-1)^{f(x)+f(x+s+s')} \\
  &= (F * F)(s + s'). \label{eq:Magic2}
\end{align}
\Prop{Bent condition} implies that $(F*F)(x) = \delta_{x,0}$, so $(Mp)_{ss'} = 0$ for all $s \ne s'$. Since $p_{\0} = 0$, \Eq{Exact M} holds and the algorithm is exact.

To prove the converse, assume there is an exact one-query quantum algorithm that solves $\BHSP_f$. Then \Eq{Exact M} holds for some $p_{\0}$ and $p$ that form a probability distribution on $\set{\0} \cup \Z_2^n$.

First, we claim that without loss of generality, the probabilities $p_x$ can be set equal for all $x \in \Z_2^n$. More precisely, we set $\bar{p} \defeq (1 - p_{\0}) \mu$ and show that \Eq{Exact M} still holds if we replace $p$ by $\bar{p}$. Note that $1 - p_{\0} = \sum_{y \in \Z_2^n} p_{x+y}$ for any $x \in \Z_2^n$, so
\begin{align}
  (M\bar{p})_{ss'}
  &= \frac{1}{2^n} \sum_{x \in \Z_2^n} M_{ss'\!\!,x} \, (1-p_{\0}) \\
  &= \frac{1}{2^n} \sum_{x \in \Z_2^n} (-1)^{f(x+s)+f(x+s')} \sum_{y \in \Z_2^n} p_{x+y} \\
  &= \frac{1}{2^n} \sum_{y \in \Z_2^n} \sum_{x \in \Z_2^n} (-1)^{f(x+y+s)+f(x+y+s')} p_x \\
  &= \frac{1}{2^n} \sum_{y \in \Z_2^n} \sum_{x \in \Z_2^n} M_{(y+s,y+s'),x} \, p_x \\
  &= \frac{1}{2^n} \sum_{y \in \Z_2^n} (Mp)_{(y+s,y+s')} \\
  &= -p_{\0}
\end{align}
where the last equality follows since $p$ is a solution of \Eq{Exact M}. We conclude that $\bar{p}$ is also a solution of \Eq{Exact M}, i.e.,
\begin{equation}
  (1-p_{\0}) M \mu = -p_{\0} u.
  \label{eq:Uniform solution}
\end{equation}

Recall from Eqs.~\EqRef{Magic1} to~\EqRef{Magic2} that $(M \mu)_{ss'} = (F*F)(s+s')$, which together with \Eq{Uniform solution} implies that $(1 - p_{\0}) (F*F)(s+s') = -p_{\0}$ for all $s \neq s'$. Clearly, there is no solution with $p_{\0} = 1$. Thus we have
\begin{equation}
  (F*F)(w) = - \frac{p_{\0}}{1-p_{\0}} \leq 0
  \label{eq:F*F}
\end{equation}
for any $w \neq 0$. Observe that $(F*F)(w) = \sum_{x \in \Z_2^n} \frac{1}{2^n} (-1)^{f(x)+f(x+w)}$ is an integer multiple of $1/2^n$ and $(F*F)(0) = 1$ for any $f$. Thus, we can rewrite \Eq{F*F} as
\begin{equation}
  (F*F)(w) =
  \begin{cases}
    1 & \text{if $w = 0$,} \\
    -k/2^n & \text{otherwise}
  \end{cases}
  \label{eq:F*F cases}
\end{equation}
for some integer $k \geq 0$. Therefore
\begin{equation}
  \sum_{w \in \Z_2^n} (F*F)(w) = 1 - \frac{2^n - 1}{2^n} k.
  \label{eq:Discrete}
\end{equation}
On the other hand,
\begin{align}
  \sum_{w \in \Z_2^n} (F*F)(w)
  &= \sum_{w \in \Z_2^n} \sum_{x \in \Z_2^n} \frac{1}{2^n} (-1)^{f(x)+f(x+w)} \\
  &= \biggl[ \frac{1}{\sqrt{2^n}} \sum_{x \in \Z_2^n} (-1)^{f(x)} \biggr]^2 \\
  &= \frac{1}{2^n} \biggl[ \, \sum_{x \in \Z_2^n} \bigl( 1 - 2 f(x) \bigr) \biggr]^2 \\
  &= \frac{1}{2^n} \bigl( 2^n - 2 \abs{f} \bigr)^2.
  \label{eq:Weight}
\end{align}
Putting this together with \Eq{Discrete} gives
\begin{equation}
  \bigl( 2^n - 2 \abs{f} \bigr)^2 = 2^n - (2^n - 1) k.
  \label{eq:|f| and k}
\end{equation}
This equation has no solutions for $k \geq 2$ since the right-hand side is negative (for $n \geq 2$). Similarly, there are no solutions for $k = 1$ since the left-hand side is even and the right-hand side is odd. Therefore $k = 0$ (and hence $p_{\0} = 0$), which implies that $f$ is bent by \Eq{F*F cases} and \Prop{Bent condition}.
\end{proof}

Note that there is a solution to \Eq{|f| and k} with $k=2$ and $n=1$, provided $\abs{f} = 1$.  This trivial case involves the one-argument Boolean functions $f(x) = x$ and $f(x) = \NOT(x)$.  For these functions which we can choose $p_{\0} = 1/2$ and $p_0 = p_1 = 1/4$ to determine the hidden shift exactly with one query.  A deterministic classical algorithm can also solve $\BHSP_f$ with one query for these functions.

%%%%%%%%%%%%%%%%%%%%%%%%%%%%%%%%%%%%%%%%%%%%%%%%%%%%%%%
\section{Success probability of one-query PGM for random functions} \label{apx:Random1}
%%%%%%%%%%%%%%%%%%%%%%%%%%%%%%%%%%%%%%%%%%%%%%%%%%%%%%%

In this appendix, we show that for one query, the expected success probability of $\PGM(f,1)$ approaches a constant less than $1$ for large $n$. This suggests that one query might not be enough to solve the problem with success probability arbitrarily close to $1$. However, we do not know if the PGM algorithm has optimal success probability in the one-query case.

\begin{theorem}\label{thm:Random1}
Let $f$ be an $n$-argument Boolean function chosen uniformly at random and suppose that a hidden shift for $f$ is chosen adversarially. Then \PGM$(f,1)$ solves $\BHSP_f$ with one query to $O_{f_{s}}$ and expected success probability $\bar{p} \geq 1/2$ over the choice of $f$. Indeed, $\bar{p} \ge 2/\pi - o(1)$ as $n \to \infty$.
\end{theorem}

\begin{proof}
Recall from \Eq{avg t-fold coefficient} in \Lem{PGM} that \PGM$(f,t)$ recovers the hidden shift of $f$ correctly after $t$ queries with success probability $p_f(t)$. If the function $f$ is chosen uniformly at random, then the expected success probability after $t$ queries is
\begin{equation}
  \bar{p}(t)
  \defeq \frac{1}{2^{2^n}} \sum_f p_f(t)
   = \frac{1}{2^{2^n}} \sum_f \frac{1}{2^n}
     \Biggl( \sum_{w \in \Z_2^n} \FC{t}(w) \Biggr)^2.
\end{equation}
We can obtain a lower bound on $\bar{p}(t)$ 
using the Cauchy-Schwarz inequality:
\begin{equation}
  \bar{p}(t)
  \geq \frac{1}{2^n} \frac{1}{(2^{2^n})^2}
    \Biggl( \sum_f \sum_{w \in \Z_2^n} \FC{t}(w) \Biggr)^2 \!\!\!
  = 2^n
    \Biggl(
      \frac{1}{2^n} \sum_{w \in \Z_2^n}
      \frac{1}{2^{2^n}} \sum_f \FC{t}(w)
    \Biggr)^2 \!\!\! \equalscolon \tilde{p}(t).
  \label{eq:bar-p and tilde-p}
\end{equation}
Taking $t = 1$, this gives
\begin{align}
  \bar{p}
  &\geq
     \frac{1}{2^n} \frac{1}{(2^{2^n})^2}
     \Biggl( \sum_f \sum_{w \in \Z_2^n} \abs{\hat{F}(w)} \Biggr)^2 \\
  &= \frac{1}{2^n}
     \Biggl(
       \frac{1}{2^{2^n}}
       \sum_{w \in \Z_2^n} \sum_f \,
       \abs[\bigg]{\frac{1}{2^n} \sum_{x \in \Z_2^n} (-1)^{w \cdot x + f(x)}}
     \Biggr)^2.
\end{align}
For each $w$ we can define $f'(x) \defeq w \cdot x + f(x)$ and change the order of summation by summing over $f'$ instead of $f$. The value of this sum does not depend on $w$, so we get
\begin{equation}
  \bar{p}
  \geq \frac{1}{2^n}
       \left(
         \frac{1}{2^{2^n}} \sum_f \,
         \abs[\Bigg]{\sum_{x \in \Z_2^n} (-1)^{f(x)}}
       \right)^2
= \frac{L(2^n)^2}{2^n}
  \label{eq:pbar and L}
\end{equation}
where
\begin{equation}
  L(N) \defeq
    \frac{1}{2^N} \!\!
    \sum_{z \in \set{1,-1}^N} \, \abs[\Bigg]{\sum_{i=1}^N z_i}
\end{equation}
is the expected distance traveled by $N$ steps of a random walk on a line (where each step is of size one and is to the left or the right with equal probability). It remains to lower bound $L(N)$.

Let $N = 2m$ for some integer $m \geq 1$. Using standard identities for sums of binomial coefficients, we compute
\begin{align}
  L(2m)
   &= \frac{1}{2^{2m}} \cdot 2 \sum_{k=0}^m (2m-2k) \binom{2m}{k} \\
   &= \frac{1}{2^{2m}} \cdot 2 m \binom{2m}{m}.
\end{align}
Since the central binomial coefficient satisfies~\cite[p.~48]{Koshy}
\begin{equation}
  \binom{2m}{m} \geq \frac{4^m}{\sqrt{4m}},
\end{equation}
we find
\begin{equation}
  L(2m) \geq \sqrt{m}.
\end{equation}
For $N = 2^n$ this gives $L(2^n) \geq \sqrt{2^n/2}$. We plug this in \Eq{pbar and L} and get $\bar{p} \geq 1/2$. In fact, according to Stirling's formula $\binom{2m}{m} \sim 4^m / \sqrt{\pi m}$ as $m \to \infty$. This means that $L(N) \sim \sqrt{2N / \pi}$ as $N \to \infty$ and our lower bound on $\bar{p}$ approaches $2/\pi$ as $n \to \infty$.
\end{proof}

%%%%%%%%%%%%%%%%%%%%%%%%%%%%%%%%%%%%%%%%%%%%%%%%%%
\section{Two queries suffice for random functions} \label{apx:Random functions}
%%%%%%%%%%%%%%%%%%%%%%%%%%%%%%%%%%%%%%%%%%%%%%%%%%

In this appendix we prove the following:
\RANDOM*

%-------------------%
\subsection{Strategy}
%-------------------%

Our goal is lower bound $\tilde{p}(t)$, as defined in \Eq{bar-p and tilde-p}. Let us define a random variable $X$ over Boolean functions $f\colon \Z_2^n \to \Z_2$ and binary strings $w \in \Z_2^n$, whose value is
\begin{equation}
  X \defeq \bigl[ \FC{t}(w) \bigr]^2
     = \bigl[ \hat{F}^2 \bigr]^{*t} (w),
\end{equation}
where $f$ and $w$ are chosen uniformly at random. Notice from \Eq{bar-p and tilde-p} that
\begin{equation}
  \tilde{p}(t) = 2^n \bigl( \E[\sqrt{X}] \bigr)^2.
  \label{eq:tilde-p and EsqrtX}
\end{equation}
Clearly, for any $x \geq 0$ we have
\begin{equation}
  \E[\sqrt{X}] \geq \sqrt{x} \, \Pr(X \geq x).
  \label{eq:Cutoff}
\end{equation}
Our strategy is to use a one-sided version of Chebyshev's inequality, known as Cantelli's inequality, to lower-bound $\Pr(X \geq x)$, and then choose a value of $x$ that maximizes our lower bound on $\tilde{p}(t)$.

\begin{fact*}[Cantelli's inequality]\label{fact:Cantelli}
Let $\mu \defeq \E[X]$ and $\sigma^2 \defeq \E[X^2] - \mu^2$ be the mean and variance of $X$, respectively. Then $\Pr(X - \mu \geq k \sigma) \geq \frac{1}{1 + k^2}$.
\end{fact*}

Alternatively, if we substitute $X$ by $-X$ and reverse the inequality then
\begin{equation}
  \Pr(X \geq \mu - k \sigma) \geq \frac{k^2}{1 + k^2}.
\end{equation}
If we substitute $x \defeq \mu - k \sigma$ in \Eq{Cutoff}, then according to the above inequality,
\begin{equation}
  \E[\sqrt{X}] \geq \sqrt{\mu - k \sigma} \, \frac{k^2}{1 + k^2}.
  \label{eq:EsqrtX}
\end{equation}
Using \Eq{bar-p and tilde-p}, \Eq{tilde-p and EsqrtX}, and \Eq{EsqrtX} gives
\begin{equation}
  \bar{p}(t)
   \geq \tilde{p}(t)
      = 2^n \bigl( \E[\sqrt{X}] \bigr)^2
   \geq 2^n (\mu - k \sigma) \biggl( 1 + \frac{1}{k^2} \biggr)^{-2}.
  \label{eq:bar-p bound}
\end{equation}
It remains to lower bound $\mu$ (\Sect{Mean}), upper bound $\sigma$ (\Sect{Variance}), and make a reasonable choice of the deviation parameter $k$ (\Sect{Deviation}).

%-----------------------------%
\subsection{Computing the mean} \label{sect:Mean}
%-----------------------------%

Let us compute the mean
\begin{equation}
  \mu = \E[X]
      = \frac{1}{2^{2^n}} \sum_f \frac{1}{2^n} \sum_{w \in \Z_2^n}
        \bigl[ \hat{F}^2 \bigr]^{*t} (w)
\end{equation}
for any integer $t \geq 1$. Notice that
\begin{align}
   \sum_{w \in \Z_2^n} \bigl[ \hat{F}^2 \bigr]^{*t} (w) \;
&= \sum_{w,y_1, \dotsc, y_{t-1} \in \Z_2^n}
   \hat{F}(y_1)^2 \dotsb \hat{F}(y_{t-1})^2 \hat{F} \bigl( w - (y_1 + \dotsb + y_{t-1}) \bigr)^2 \\
&= \;\;\;\sum_{y_1, \dotsc, y_t \in \Z_2^n}
   \hat{F}(y_1)^2 \dotsb \hat{F}(y_{t-1})^2 \hat{F}(y_t)^2 \\
&= \Biggl( \sum_{y \in \Z_2^n} \hat{F}(y)^2 \Biggr)^t \\
&= 1
\end{align}
by unitarity of the Fourier transform (see Plancherel's identity in \Sect{QFT and convolution}). We conclude that
\begin{equation}
  \mu = \frac{1}{2^n}
  \label{eq:mu}
\end{equation}
independent of $t$.

%---------------------------------%
\subsection{Computing the variance} \label{sect:Variance}
%---------------------------------%

Next we compute the variance
\begin{equation}
  \E[X^2] = \frac{1}{2^{2^n}} \sum_f \frac{1}{2^n} \sum_{w \in \Z_2^n}
            \Bigl( \bigl[ \hat{F}^2 \bigr]^{*t} (w) \Bigr)^2.
  \label{eq:X2 initial}
\end{equation}
Note that from \Eq{FS norm} and Plancherel identity we have
\begin{equation}
    \sum_{w \in \Z_2^n} \Bigl( \bigl[ \hat{F}^2 \bigr]^{*t} (w) \Bigr)^2
  = \sum_{w \in \Z_2^n}
      \Biggl( \frac{1}{\sqrt{2^n}}
        \widehat{\:(F*F)^t\:} (w)
      \Biggr)^2
  = \frac{1}{2^n} \sum_{w \in \Z_2^n} (F*F)^{2t} (w).
\end{equation}
We substitute this in \Eq{X2 initial} and get
\begin{align}
  \E[X^2]
  &= \frac{1}{2^{2^n}} \sum_f \frac{1}{2^n}
     \Biggl( \frac{1}{2^n} \sum_{w \in \Z_2^n} (F*F)^{2t} (w) \Biggr) \\
  &= \frac{1}{2^{2n}} \sum_{w \in \Z_2^n}
     \frac{1}{2^{2^n}} \sum_f \Biggl( \frac{1}{2^n} \sum_{x \in \Z_2^n} (-1)^{f(x) + f(w+x)} \Biggr)^{2t}.
  \label{eq:X2}
\end{align}

%.............................................%
\subsubsection{Counting pairings} \label{sect:Pairings}
%.............................................%

Let us introduce some combinatorial ideas that will help us to evaluate the sum in \Eq{X2}.
\begin{definition}
Let $S$ be a finite set and let $l \geq 1$ be an integer. We say that $a_1, a_2, \dotsc, a_{2l} \in S$ are \emph{paired} if there exists a permutation $\pi$ of $\set{1, 2, \dotsc, 2l}$ such that $a_{\pi(2i-1)} = a_{\pi(2i)}$ for all $i \in \set{1, 2, \dotsc, l}$. Define $\Delta\colon S^{2l} \to \Z_2$ as
\begin{equation}
  \Delta(a_1, a_2, \dotsc, a_{2l}) \defeq
  \begin{cases}
    1 & \text{if $a_1, a_2, \dotsc, a_{2l}$ are paired}, \\
    0 & \text{otherwise}.
  \end{cases}
\end{equation}
\end{definition}

Notice that for $l=2$ we have
$\Delta(a,b,c,d)
  = \delta_{a,b} \delta_{c,d}
  + \delta_{a,c} \delta_{b,d}
  + \delta_{a,d} \delta_{b,c}
  - 2 \delta_{a,b,c,d}$,
so the number of ways to pair four elements of $S$ is
\begin{equation}
  \sum_{a,b,c,d \in S} \!\! \Delta(a,b,c,d)
  \;=\; 3 \!\! \sum_{a,b,c,d \in S} \!\! \delta_{a,b} \delta_{c,d}
  \;-\; 2 \!\! \sum_{a,b,c,d \in S} \!\! \delta_{a,b,c,d}
  \;=\; 3 \abs{S}^2 - 2 \abs{S}.
  \label{eq:sum abcd}
\end{equation}

\begin{proposition}
Let $S = \set{0,1}^n$. Then
for any $a_1, a_2, \dotsc, a_{2l} \in S$,
\begin{equation}
  \frac{1}{2^{2^n}} \sum_f (-1)^{f(a_1) + f(a_2) + \dotsb + f(a_{2l})}
  = \Delta(a_1, a_2, \dotsc, a_{2l})
  \label{eq:Boolean pairings}
\end{equation}
where the sum is over all Boolean functions $f\colon \Z_2^n \to \Z_2$.
\end{proposition}

\begin{proof}
Clearly, if $a_1, a_2, \dotsc, a_{2l}$ are paired, then the exponent of $-1$ is even and the sum is $1$. Otherwise, we can omit the paired arguments, and all remaining $a_i$ are distinct. Since we are averaging over all $f$ and the values that $f$ takes at distinct points are independent, the sum vanishes.
\end{proof}

We can use this observation to rewrite \Eq{X2} as follows:
\begin{equation}
  \E[X^2]
   = \frac{1}{2^{2(t+1)n}} \sum_{w \in \Z_2^n}
     \sum_{a_1, \dotsc, a_{2t} \in \Z_2^n}
     \Delta(a_1, a_1 + w, a_2, a_2 + w, \dotsc, a_{2l}, a_{2l} + w).
  \label{eq:X2 combinatorial}
\end{equation}

%....................................................................%
\subsubsection{Evaluating the variance at \texorpdfstring{$t=2$}{t=2}} \label{sect:Variance t=2}
%....................................................................%

In general, the variance depends on $t$. However, we are interested only in the $t = 2$ case, so from now on we will assume that $t = 2$ and do not write the dependence on $t$ explicitly.  For $t=2$, \Eq{X2 combinatorial} reads
\begin{equation}
  \E[X^2]
  = \frac{1}{2^{6n}} \sum_{w \in \Z_2^n}
    \sum_{a,b,c,d \in \Z_2^n}
    \Delta(a, a + w, b, b + w, c, c + w, d, d + w).
  \label{eq:EX2}
\end{equation}

We consider two cases. First, when $w = 0$, the eight arguments of $\Delta$ are always paired, so the inner sum in \Eq{EX2} evaluates to
\begin{equation}
  \sum_{a,b,c,d \in \Z_2^n}
  \Delta(a, a, b, b, c, c, d, d)
  = 2^{4n}.
  \label{eq:w=0}
\end{equation}

Now suppose $w \neq 0$. Then $w_i = 1$ for some $i \in \set{1, \dotsc, n}$ and thus either $a_i = 0$ or $a_i + w_i = 0$ (and similarly for $b$, $c$, and $d$). In total there are $2^4 = 16$ cases. Since $\Delta$ is invariant under permutations of arguments, we can substitute $a$ by $a + w$, which effectively swaps the arguments $a$ and $a + w$. By performing a similar operation for $b$, $c$, and $d$, we can ensure that $a_i = b_i = c_i = d_i = 0$. Among the eight arguments of $\Delta$ in \Eq{EX2}, arguments $a$, $b$, $c$, and $d$ can be paired only among themselves since $w_i = 1$. Moreover, once $a$ and $b$ are paired, then so are $a + w$ and $b + w$. Thus, we can restrict the $i$th bit of $w$ to be $1$ and ignore the four extra arguments of $\Delta$. Then the inner sum in \Eq{EX2} becomes
\begin{align}
  16 \sum_{a,b,c,d \in \Z_2^{n-1}} \Delta(a, b, c, d)
= 16 \cdot \bigl( 3 \cdot 2^{2n-2} - 2 \cdot 2^{n-1} \bigr)
= 12 \cdot 2^{2n} - 16 \cdot 2^n,
  \label{eq:w!=0}
\end{align}
where the first equality follows from \Eq{sum abcd} with $S = \Z_2^{n-1}$.

By combining \Eq{w=0} and \Eq{w!=0}, we can rewrite \Eq{EX2} as
\begin{align}
  \E[X^2]
 &= \frac{1}{2^{6n}} \biggl(
      2^{4n} + (2^n - 1) \cdot (12 \cdot 2^{2n} - 16 \cdot 2^n)
    \biggr) \\
 &= \frac{1}{2^{2n}}
 + \frac{12}{2^{3n}}
 - \frac{28}{2^{4n}}
 + \frac{16}{2^{5n}}.
\end{align}
Using the value of $\mu$ from \Eq{mu}, we see that for $n \geq 1$ the variance is
\begin{equation}
   \sigma^2
   = \E[X^2] - \mu^2
   = \frac{12}{2^{3n}}
   - \frac{28}{2^{4n}}
   + \frac{16}{2^{5n}}
\geq \frac{1}{2^{3n}}. \label{eq:sigma2}
\end{equation}

%---------------------------------%
\subsection{Choosing the deviation} \label{sect:Deviation}
%---------------------------------%

To complete the lower bound on the success probability, recall from \Eq{bar-p bound} that
\begin{equation}
  \bar{p} \geq 2^n (\mu - k \sigma) \biggl( 1 + \frac{1}{k^2} \biggr)^{-2}.
\end{equation}
Substituting the bounds on $\mu$ and $\sigma$ from \Eq{mu} and \Eq{sigma2}, respectively, gives
\begin{equation}
  \bar{p}
  \geq \biggl( 1 - \frac{k}{\sqrt{2^n}} \biggr)
       \biggl( 1 + \frac{1}{k^2} \biggr)^{-2}.
\end{equation}
Notice that $\left( 1 + \frac{1}{k^2} \right)^{-2} \geq 1 - \frac{2}{k^2}$ for any $k$, so
\begin{equation}
  \bar{p}
  \geq \biggl( 1 - \frac{k}{\sqrt{2^n}} \biggr)
       \biggl( 1 - \frac{2}{k^2} \biggr)
  \geq 1 - \frac{k}{\sqrt{2^n}} - \frac{2}{k^2}.
\end{equation}

It remains to make a good choice for $k$. Let $\alpha = \sqrt{2^n}$ and $k = \alpha^c$ for some $c > 0$. Then
\begin{align}
  \bar{p} \geq 1 - \alpha^{c-1} - 2 \alpha^{-2c}.
\end{align}
Choosing $c = 1/3$ (i.e., $k = 2^{n/6}$) gives
\begin{equation}
  \bar{p} \geq 1 - \frac{3}{64} \cdot 2^{-n}.
\end{equation}
This concludes the proof of \Thm{Random2}.

%%%%%%%%%%%%%%%%%%%%%%%%%%%%%%%%%%%%%%%%
\section{Zeroes in the Fourier spectrum} \label{apx:Zeroes}
%%%%%%%%%%%%%%%%%%%%%%%%%%%%%%%%%%%%%%%%

%----------------------------------------------%
\subsection{Undetectable shifts and anti-shifts} \label{sect:b-shifts}
%----------------------------------------------%

In some cases the Boolean hidden shift problem cannot be solved exactly in principle. For example, if the function $f$ is invariant under some shift, then the hidden shift cannot be uniquely determined, as the oracle does not contain enough information (an extreme case of this is a constant function which is invariant under all shifts). In this section we consider such degenerate functions and analyze their Fourier spectra.

\begin{definition}
Let $b \in \Z_2$. We say that $\bv{s}$ is a \emph{$b$-shift} for a function $f\colon \Z_2^n \to \Z_2$ if $f$ has the following property: $\forall \bv{x} \in \Z_2^n\colon f(\bv{x} + \bv{s}) = f(\bv{x}) + b$. We refer to $0$-shifts as \emph{undetectable shifts} since they cannot be distinguished from the trivial shift $\bv{s} = 0$. We also refer to $1$-shifts as \emph{anti-shifts} since they negate the truth table of $f$.
\end{definition}

The following result provides an alternative characterization of $b$-shifts. It relates the maximal and minimal autocorrelation value of $F$ to undetectable shifts and anti-shifts of $f$, respectively (see \Def{Convolution} for the definition of convolution).

\begin{proposition}
The string $\bv{s} \in \Z_2^n$ is a $b$-shift for function $f\colon \Z_2^n \to \Z_2$ if and only if $(F*F)(\bv{s}) = (-1)^b$, where $F(\bv{x})\defeq(-1)^{f(\bv{x})} / \sqrt{2^n}$ for all $\bv{x} \in \Z_2^n$.
\end{proposition}

\begin{proof}
Let $\bv{s}$ be a $b$-shift of $f$. Then
\begin{align}
  (F*F)(\bv{s})
  &=               \sum_{\bv{x} \in \Z_2^n} F(\bv{x}) F(\bv{x} + \bv{s}) \label{eq:Self correlation} \\
  &= \frac{1}{2^n} \sum_{\bv{x} \in \Z_2^n} (-1)^{f(\bv{x})} (-1)^{f(\bv{x}) + b} \\
  &= \frac{1}{2^n} \sum_{\bv{x} \in \Z_2^n} (-1)^b \\
  &= (-1)^b.
\end{align}

For the converse, note that all terms on the right-hand side of \Eq{Self correlation} have absolute value equal to $1/2^n$. In total there are $2^n$ terms, so $\abs{(F*F)(\bv{s})} \leq 1$. If this bound is saturated, then all terms in \Eq{Self correlation} must have the same phase. Thus, $\bv{s}$ is a $b$-shift for some $b \in \Z_2$.
\end{proof}

If $\bv{s}'$ and $\bv{s}''$ are undetectable shifts of $f$ then so is $\bv{s}' + \bv{s}''$, since $f(\bv{x} + \bv{s}' + \bv{s}'') = f(\bv{x} + \bv{s}') = f(\bv{x})$ for any $\bv{x}$. Hence the set of all undetectable shifts forms a linear subspace of $\Z_2^n$. Also, if $\bv{a}'$ and $\bv{a}''$ are anti-shifts, then $\bv{a}' + \bv{a}''$ is an undetectable shift. In particular, a Boolean function with no undetectable shifts can have at most one anti-shift.

If we want to solve the hidden shift problem for a function $f$ that has an undetectable shift $s$, we can apply an invertible linear transformation $A$ on the input variables such that $A \cdot 0 \dots 01 = s$. Thus we simulate the oracle for the function $f'(x) \defeq f(A \cdot x)$ such that $f'(x + 0 \dots 01) = f'(x)$. Notice that $f'$ is effectively an $(n-1)$-argument function, since it does not depend on the last argument. Similarly, if $f$ has a $k$-dimensional subspace of undetectable shifts, it is effectively an $(n-k)$-argument function. Solving the hidden shift problem for such a function is equivalent to solving it for the reduced $(n-k)$-argument function $f'$ and picking arbitrary values for the remaining $k$ arguments. In this sense, Boolean functions with undetectable shifts are degenerate and we can consider only functions with no undetectable shifts without loss of generality.

Similarly, if $f$ has an anti-shift, we can use the same construction to show that it is equivalent to a function $f'$ such that $f'(x_1, \dotsc, x_{n-1}, x_n) = f''(x_1, \dotsc, x_{n-1}) \op x_n$ where $f''$ is an $(n-1)$-argument function. To solve the hidden shift problem for $f'$, we first solve it for $f''$ and then learn the value of the remaining argument $x_n$ via a single query. In this sense, Boolean functions with anti-shifts are also degenerate. Thus, without loss of generality we can consider the hidden shift problem only for \emph{non-degenerate} functions, i.e., ones that have no $b$-shifts for any $b \in \Z_2$.

Finally, let us show that Boolean functions with $b$-shifts have at least half of their Fourier coefficients equal to zero. Let $\mc{S}$ be an $(n-1)$-dimensional subspace of $\Z_2^n$, and let us denote the two cosets of $\mc{S}$ in $\Z_2^n$ by $\mc{S}_b \defeq \mc{S} + b \bv{r}$, where $b \in \Z_2$ and $\bv{r} \in \Z_2^n \setminus \mc{S}$ is any representative of the coset for $b=1$. The following result relates the property of having a $b$-shift to the property of having zero Fourier coefficients with special structure.

\begin{lemma}\label{lem:b-shifts}
A function $f\colon \Z_2^n \to \Z_2$ has a non-zero $b$-shift if and only if there is an $(n-1)$-dimensional subspace $\mc{S} \subset \Z_2^n$ such that $\hat{F}(\bv{w}) = 0$ when $\bv{w} \notin \mc{S}_b$.
\end{lemma}

\begin{proof}
Assume that $\bv{s}$ is a $b$-shift of $f$. Then
\begin{align}
  \hat{F}(\bv{w})
  &= \frac{1}{2^n} \sum_{\bv{x} \in \Z_2^n} (-1)^{\bv{w} \cdot \bv{x} + f(\bv{x})} \\
  &= \frac{1}{2^n} \sum_{\bv{x} \in \Z_2^n} (-1)^{\bv{w} \cdot (\bv{x}+\bv{s}) + f(\bv{x}+\bv{s})} \\
  &= \frac{1}{2^n} \sum_{\bv{x} \in \Z_2^n} (-1)^{\bv{w} \cdot (\bv{x}+\bv{s}) + f(\bv{x}) + b} \\
  &= (-1)^{\bv{w} \cdot \bv{s} + b} \frac{1}{2^n} \sum_{\bv{x} \in \Z_2^n}
     (-1)^{\bv{w} \cdot \bv{x} + f(\bv{x})} \\
  &= (-1)^{\bv{w} \cdot \bv{s} + b} \hat{F}(\bv{w}).
\end{align}
Thus, $\hat{F}(\bv{w}) = 0$ when $\bv{w} \cdot \bv{s} \neq b$. Let $\mc{S}$ be the $(n-1)$-dimensional subspace of $\Z_2^n$ orthogonal to $\bv{s}$. Then $\bv{w} \in \mc{S}_b \Leftrightarrow \bv{w} \cdot \bv{s} = b$ and thus $\hat{F}(\bv{w}) = 0$ when $\bv{w} \notin \mc{S}_b$.

For the converse, assume that $\mc{S}$ is an $(n-1)$-dimensional subspace of $\Z_2^n$ and $\hat{F}(\bv{w}) = 0$ when $\bv{w} \notin \mc{S}_b$. Let $\bv{s} \in \Z_2^n$ be the unique non-zero vector orthogonal to $\mc{S}$. Then $\mc{S}_b = \set{\bv{w} \colon \bv{w} \cdot \bv{s} = b}$ and we have
\begin{align}
  F(\bv{x} + \bv{s})
  &= \hat{\hat{F}}(\bv{x} + \bv{s}) \\
  &= \frac{1}{\sqrt{2^n}} \sum_{\bv{w} \in \Z_2^n} (-1)^{(\bv{x}+\bv{s}) \cdot \bv{w}} \hat{F}(\bv{w}) \\
  &= \frac{1}{\sqrt{2^n}} \sum_{\bv{w} \in \mc{S}_b} (-1)^{(\bv{x}+\bv{s}) \cdot \bv{w}} \hat{F}(\bv{w}) \\
  &= (-1)^b \frac{1}{\sqrt{2^n}} \sum_{\bv{w} \in \mc{S}_b} (-1)^{\bv{x} \cdot \bv{w}} \hat{F}(\bv{w}) \\
  &= (-1)^b F(\bv{x}).
\end{align}
Hence $f(\bv{x} + \bv{s}) = f(\bv{x}) + b$ and thus $\bv{s}$ is a $b$-shift of $f$.
\end{proof}

%-------------------------%
\subsection{Decision trees} \label{sect:Trees}
%-------------------------%

In the previous section we discussed degenerate cases of Boolean functions that have many zero Fourier coefficients. In this section we explain how to construct non-degenerate examples.

\begin{lemma}\label{lem:Tree}
If $f$ is a Boolean function defined by a decision tree of height $h$ then $\hat{F}(w) = 0$ when $\abs{w} > h$.
\end{lemma}

\begin{proof}
Since the Boolean function $f$ is given by a decision tree, let $\set{P_1, \dotsc, P_m}$ be the set of all paths that start at the root of this tree and end at a parent of a leaf labeled by $1$. For example, $P_1 = \set{x_2, x_1, x_5, x_4, x_{10}}$ and $P_2 = \set{x_2, x_7, x_1}$ are two such paths for the tree shown in \Fig{Tree}. We can write the disjunctive normal form of $f$ as
\begin{equation}
  f(x) = \bigvee_{i=1}^m \bigwedge_{j \in P_i} \bigl( b^{(i)}_j \op x_j \bigr)
  \label{eq:DNF}
\end{equation}
where ``$\vee$'' and ``$\wedge$'' represent logical OR and AND functions, respectively, and $b^{(i)}_j \in \Z_2$ is equal to $1$ if and only if variable $x_j$ has to be negated on path $P_i$. For example, $x_{10}$ is negated on $P_1$, and $x_2$ and $x_7$ are negated on $P_2$.

To prove the desired result about the Fourier coefficients of $f$, we switch from Boolean functions to $(\pm 1)$-valued functions with $(\pm 1)$-valued variables. In particular, we replace $f\colon \Z_2^n \to \Z_2$ by a function $\tilde{F}\colon \set{1,-1}^n \to \set{1,-1}$ in variables $X_i \in \set{1,-1}$ such that
\begin{equation}
  \tilde{F} \bigl( (-1)^x \bigr) = (-1)^{f(x)}
  \label{eq:tilde F}
\end{equation}
for all $x \in \Z_2^n$.

Notice that the $(\pm 1)$-valued versions of logical NOT, AND, and OR functions are given by the following polynomials:
\begin{align}
  \NOT(X) &\defeq - X, \\
  \AND(X_1, \dotsc, X_k) &\defeq \phantom{+} 1 - 2 \prod_{i=1}^k \frac{1 - X_i}{2}, \\
   \OR(X_1, \dotsc, X_k) &\defeq          -  1 - 2 \prod_{i=1}^k \frac{1 + X_i}{2}. \label{eq:OR}
\end{align}
We can use these polynomials and \Eq{DNF} to write $\tilde{F}$ as
\begin{equation}
  \tilde{F}(X) = \OR_{i=1}^m \AND_{j \in P_i} (-1)^{b^{(i)}_j} X_j,
  \label{eq:F DNF}
\end{equation}
where $\OR_{i=1}^m X_i$ stands for $\OR(X_1, \dotsc, X_m)$ and a similar convention is used for AND.

When we determine the value of $f$ using a decision tree, each input $x \in \Z_2^n$ leads to a unique leaf of the tree. Thus, when $f(x) = 1$, there is a unique value of $i$ in \Eq{DNF} for which the corresponding term in the disjunction is satisfied. With this promise we can simplify \Eq{OR} to
\begin{equation}
  \OR(X_1, \dotsc, X_k) \defeq \sum_{i=1}^k (X_i - 1) + 1.
\end{equation}
If we use this in \Eq{F DNF}, we get
\begin{align}
  \tilde{F}(X)
  &= \sum_{i=1}^m \Bigl( \AND_{j \in P_i} (-1)^{b^{(i)}_j} X_j - 1 \Bigr) + 1, \\
  &= 1 - 2 \sum_{i=1}^m \prod_{j \in P_i} \frac{1 - (-1)^{b^{(i)}_j} X_j}{2}. \label{eq:F(X) poly}
\end{align}
Notice that this polynomial has degree at most $\max_i \abs{P_i} \leq h$, the height of the tree. On the other hand, the Fourier transform is self-inverse (see \Sect{QFT and convolution}), so
\begin{equation}
  (-1)^{f(x)}
  = \sqrt{2^n} F(x)
  = \sqrt{2^n} \hat{\hat{F}}(x)
  = \sum_{w \in \Z_2^n} (-1)^{x \cdot w} \hat{F}(w).
\end{equation}
The $(\pm 1)$-valued equivalent of this equation is
\begin{equation}
  \tilde{F}(X)
  = \sum_{w \in \Z_2^n} \hat{F}(w) \prod_{i \colon w_i = 1} X_i.
\end{equation}
By comparing this with \Eq{F(X) poly} we conclude that $\hat{F}(w) = 0$ when $\abs{w} > h$.
\end{proof}

According to this lemma, we can use the following strategy to construct Boolean functions with a large fraction of their Fourier coefficients equal to zero. We pick a random decision tree with many variables but small height, i.e., large $n$ and small $h$ (notice that $n \leq 2^h - 1$). Then we are guaranteed that the fraction of non-zero Fourier coefficients does not exceed
\begin{equation}
  \frac{1}{2^n} \sum_{k=0}^{h} \binom{n}{k}
  \leq \frac{2^{H(\frac{h}{n}) n}}{2^n}
  = \biggl( \frac{1}{2^n} \biggr)^{1-H(\frac{h}{n})}
\end{equation}
where $H(p) \defeq - p \log_2 p - (1-p) \log_2 (1-p)$ is the binary entropy function. In particular, if $h \sim \log_2 n$ then this fraction vanishes as $n$ goes to infinity, i.e., $\hat{F}$ is zero almost everywhere.

However, notice that when the number of zero Fourier coefficients is large, it is also more likely to pick a degenerate Boolean function (i.e., one that has a $b$-shift for some $b \in \Z_2$); we would like to avoid this. Recall from \Lem{b-shifts} that $f$ has a $b$-shift only if all its non-zero Fourier coefficients lie in a coset $\mc{S}_b$ of some $(n-1)$-dimensional subspace $\mc{S} \subset \Z_2^n$. Unfortunately, we do not know the probability that a random decision tree with $n$ variables and height $\log_2 n$ corresponds to a Boolean function with this property.

%---------------------------------------------------------------------------%
\subsection{Zeroes in the \texorpdfstring{$t$-fold}{t-fold} Fourier spectrum} \label{sect:Increasing t}
%---------------------------------------------------------------------------%

In this section we study the fraction of zeroes in the $t$-fold Fourier spectrum $\FC{t}$ of $f$ as a function of $t$. The main observation is \Lem{Eliminating zeroes}, which shows that unless $f$ has an undetectable shift, $\FC{t}$ becomes non-zero everywhere when $t$ is sufficiently large. This means that even for functions with a high density of zeroes in the Fourier spectrum, one can boost the success probability of the basic quantum rejection sampling approach discussed in \Sect{Basic QRS} by using the $t$-fold generalization from \Sect{t-fold QRS}.

\begin{proposition}\label{prop:expansion}
Let $S_t \defeq \set{w \in \Z_2^n \colon \FC{t}(w) \neq 0}$ be the set of strings for which $\FC{t}$ is non-zero. Then $S_{t+1} = S_t + S_1$ where $A + B \defeq \set{a+b \colon a \in A, b \in B}$.
\end{proposition}

\begin{proof}
Note that $\bigl[ \FC{t+1} \bigr]^2 \! = \bigl[ \FC{t} \bigr]^2 \! * \bigl[ \FC{1} \bigr]^2$ from \Def{Ft}. Also, $\FC{t}(w) \geq 0$ for any $t \geq 1$ and $w \in \Z_2^n$. Assume that $w_0 \in S_t$ and $w_1 \in S_1$. Then $\FC{t}(w_0) > 0$ and $\FC{1}(w_1) > 0$, so
\begin{align}
  \bigl[ \FC{t+1} \bigr]^2 (w_0 + w_1)
&    = \sum_{x \in \Z_2^n}
       \bigl[ \FC{t} \bigr]^2 (x) \cdot
       \bigl[ \FC{1} \bigr]^2 (w_0 + w_1 - x) \label{eq:Sum} \\
& \geq \bigl[ \FC{t} \bigr]^2 (w_0) \cdot
       \bigl[ \FC{1} \bigr]^2 (w_0 + w_1 - w_0) > 0.
\end{align}
Thus $w_0 + w_1 \in S_{t+1}$ and hence $S_t + S_1 \subseteq S_{t+1}$. Conversely, if $w$ cannot be written in the form $w_0 + w_1$ for some $w_0 \in S_t$ and $w_1 \in S_1$ then $\FC{t+1}(w) = 0$, since all terms of the sum in \Eq{Sum} vanish.
\end{proof}

\begin{lemma}\label{lem:Eliminating zeroes}
If $f\colon \Z_2^n \to \Z_2$ does not have an undetectable shift, then there exists $t \in \set{1, \dotsc, n}$ such that $\FC{t}$ is non-zero everywhere.
\end{lemma}

\begin{proof}
If $S_1$ spans the whole space $\Z_2^n$, we can inductively apply \Prop{expansion} to conclude that $S_t = \Z_2^n$ for some sufficiently large $t$. In particular, it suffices to take $t \leq n$ (say, if $S_1$ is the standard basis). On the other hand, if $S_1$ spans only a proper subspace of $\Z_2^n$, then it is contained in some $(n-1)$-dimensional subspace $\mc{S}_0$. Since $\FC{1} = \abs{\hat{F}}$ vanishes outside of $\mc{S}_0$, we conclude by \Lem{b-shifts} that $f$ has an undetectable shift.
\end{proof}

%%
%% Bibliography
%%

\newcommand{\STOC}[2]{Proceedings of the #1 Annual ACM Symposium on Theory of Computing (STOC #2)}
\newcommand{\FOCS}[2]{Proceedings of the #1 Annual Symposium on Foundations of Computer Science (FOCS #2)}
\newcommand{\SODA}[2]{Proceedings of the #1 ACM-SIAM Symposium on Discrete Algorithms (SODA #2)}

%\bibliography{References}

\end{document}